\documentclass{article}

     \PassOptionsToPackage{numbers, compress}{natbib}

     \usepackage[final]{neurips_2021}

\usepackage[utf8]{inputenc} 
\usepackage[T1]{fontenc}    
\usepackage{hyperref}       
\usepackage{url}            
\usepackage{booktabs}       
\usepackage{amsfonts}       
\usepackage{nicefrac}       
\usepackage{microtype}      
\usepackage{xcolor}         
\usepackage{lipsum}

\usepackage{my_style}
\usepackage{my_notation}

\title{Regulating algorithmic filtering on social media} 

\author{%
	Sarah H.~Cen\\
	MIT EECS\\
	\texttt{shcen@mit.edu} \\
	\And 
	Devavrat Shah
	\\
	MIT EECS\\
	\texttt{devavrat@mit.edu} \\
}

\begin{document}

	\maketitle


\begin{abstract}
	By filtering the content that users see, social media platforms have the ability to influence users' perceptions and decisions, from their dining choices to their voting preferences. 
	This influence has drawn scrutiny, with many calling for regulations on filtering algorithms, 
	but designing and enforcing regulations remains challenging. 
	In this work, we examine three questions. 
	First, 
	given a regulation, how would one design an audit to enforce it?
	Second, does the audit impose a performance cost on the platform?
	Third, how does the audit affect the content that the platform is incentivized to filter?
	In response, we propose a method such that, given a regulation, an auditor can test whether that regulation is met with only black-box access to the filtering algorithm. 
	We then turn to the platform's perspective. The platform's goal is to maximize an objective function while meeting regulation. 
	We find that there are conditions under which the regulation does not place a high performance cost on the platform and, notably, that content diversity can play a key role in aligning the interests of the platform and regulators.
\end{abstract}


\section{Introduction}\label{sec:intro}

In recent years,
there have been increasing calls to regulate how social media platforms
\emph{algorithmically filter} the content that appears on a user's feed. 
For example, 
one may ask that the advertisements a user sees not be based (explicitly or implicitly) on their sexual orientation  \citep{wachter2020affinity}
or that content related to public health (e.g., COVID-19) does not reflect a user's political affiliation \cite{antivax2021germani}. 

However, 
translating a regulatory guideline into an auditing procedure has proven difficult. 
Developing such a method is the focus of this work. 
As our main contribution,
we provide a general procedure such that, 
\emph{given a regulation, 
an auditor can test whether the platform complies with the regulation. }

Providing a general procedure is important because, 
without it, 
the auditing of algorithmic filtering is destined to be {reactive}: auditors must design ways to enforce a regulation as issues arise, 
and there is an inevitable delay between design and enforcement 
in order to test the proposed solutions.

On the other hand, 
designing auditing procedures is challenging because audits can have unintended side effects on the many stakeholders in the social media ecosystem.
As a result, we consider two additional questions in this work:
\begin{enumerate}[noitemsep, nosep, leftmargin=24pt]
	\item \emph{How does the proposed auditing procedure affect the platform's bottom line?}
	We consider whether the procedure imposes a high performance cost on the platform. 
	
	\item \emph{How does the auditing procedure affect the user's content?} 
	We consider what type of content the platform is incentivized to show the user when the platform complies with the regulation.
\end{enumerate}
Our main contributions are summarized as follows. 

\noindent 
\textbf{Auditing procedure}. 
As our main contribution, 
we provide a procedure such that, 
given a regulation, 
an auditor can test whether the platform complies with the regulation (Section \ref{sec:audit}). 
We restrict our attention to regulations that can be written in \emph{counterfactual form},
including the two examples given at the top of this Introduction.
Namely, the procedure applies to any regulation that can be 
written 
as:
``the filtering algorithm $\cF$ should behave similarly under inputs $\fvec$ and $\fvec'$ for all $(\fvec, \fvec') \in \cS$''.
How to quantify ``similarity'' under $\fvec$ and $\fvec'$ is crucial, 
and our second main contribution is to provide a precise notion of ``similarity'' in the context of algorithmic filtering (Section \ref{sec:problem_statement}). 

Operationally, 
the auditing procedure has several desirable properties. 
First, 
it needs only black-box access to the algorithm and therefore holds even if the filtering algorithm $\cF$ changes. 
Second, 
it does not require access to the platform's users or their personal data. 
Third, 
its parameters are interpretable and easy to tune. 
Finally, 
the procedure is modular, which allows for many possible configurations.

\noindent 
\textbf{Provable guarantees}. 
We begin by observing that algorithmic filtering is powerful (and often harmful) because \emph{information influences decisions}: 
the content that a user sees can affect how they vote, 
whether they choose to receive a vaccine, 
what restaurants they frequent, and more. 
Therefore, if one seeks to enforce a counterfactual regulation, the notion of ``similarity'' that is enforced should be with respect to the outcome of interest: the users' decisions.

Suppose that the filtering algorithm $\cF$ generates content $Z$ when given inputs $\fvec$ and $Z'$ when given $\fvec'$. 
In Section \ref{sec:theory}, 
we prove that, if $\cF$ passes the audit, 
the decision-making of \emph{any} user if shown $Z$ 
and their decision-making if they were shown $Z'$ instead are (asymptotically) indistinguishable (Theorem \ref{thm:reg_hyp_test}). 
This guarantee is powerful because it holds 
for \emph{any} user and \emph{any} decision
even though the audit does not have access to the users' personal data 
(e.g., gender) 
or know how users make decisions (e.g., how easily a user is influenced by restaurant advertisements). 
Providing such a strong guarantee {without} access to users 
is made possible by two well-known concepts from decision and learning theory: the hypothesis test and minimum-variance unbiased estimator. 
Combining these tools is one of our main technical insights and discussed in Section \ref{sec:theory}.

\noindent 
\textbf{Cost of regulation}. 
In Section \ref{sec:cost}, we study how the audit affects a platform’s ability to maximize an objective function $R$---which we refer to as \emph{reward}---and find that being audited does not necessarily place a high performance cost on the platform. 
Studying the \emph{cost of regulation} is important because
there are serious concerns that regulations 
can hurt innovation or profits, 
and our findings surface conditions under which a performance-regulation trade-off does not exist.
We further note that we leave $R$ unspecified. 
As a result, the analysis in Section \ref{sec:cost} is applicable to any $R$, 
whether it is the platform's or an advertiser's objective function. 
As examples, $R$ could measure time spent on the platform, the number of clicks on posts, or a combination of these factors.

\noindent 
\textbf{Content diversity}. 
In Section \ref{sec:diversity}, 
we turn our attention to how an audit would affect the users’ content and discover an unexpected connection to content diversity. 
We find that, under regulation, social media platforms are {incentivized} to add doses of content diversity. Put differently, when faced with a regulation, 
it is in the platform’s interest to ensure that the content it shows users is sufficiently diverse along the dimension by which $\fvec$ and $\fvec'$ for all $(\bx, \bx') \in \cS$ differ. 
Because content diversity is not a part of the regulatory test by design, this result is unexpected and suggests that content diversity plays a key role in aligning the interests of regulators and platforms.

All proofs are given in the Appendix as well as a toy example and further discussion of the audit. 

\section{Problem statement}\label{sec:problem_statement}

\subsection{System setup}

Consider a system with two agents: a social media platform and an auditor. 
\begin{enumerate}[topsep=0pt,itemsep=5pt]
	\item The \textbf{platform} selects the content that is shown to its users using a filtering algorithm $\cF : \cX \rightarrow \cZ$ 
	such that $Z = \cF (\fvec)$ is the feed produced by $\cF$ 
	given inputs $\fvec \in \cX$.
	Here, a feed is a collection of content that is shown to a user, and $\fvec$ captures any inputs that the platform uses to filter, 
	such as a user's interaction history, 
	the user's social network,
	the available content sources, 
	and so on. 
	Each feed $Z = \{ \bz_1 , \hdots , \bz_m\}$ consists of $m$ pieces of content, where $\bz_i \in \bbR^n$ for all $i \in [m]$.\footnote{
		Using a vector $\bz$ to represent each piece of content is not an addiitonal assumption because, if $\bz$ did not exist, then the platform would not be able to filter algorithmically.
	} 
	We assume that $\cF$ can be written as a {generative model}
	in which
	$\cF$ generates the content in $Z$ by drawing $m$ samples from a distribution $p_{\bz} ( \cdot \hspace{1pt} ; \genmodel ( \bx ) )$, 
	where $\genmodel (\bx) \in \parFam \in \bbR^r$ is unknown.\footnote{
		This representation is without loss of generality.
		For example, any deterministic mapping from $\fvec$ to $Z$ can be achieved by
		letting $r = nm$, $\genmodel(\bx) = (\bz_1^\top , \hdots , \bz_m^\top)$, and sequentially generating $\bz_i$ from the entries of $\genmodel(\bx)$. 
	}
	
	\item The \textbf{auditor} 
	is given a regulatory guideline that they wish to enforce. 
	The auditor's goal is to 
	check
	whether the platform's filtering algorithm $\cF$ is in compliance with the given regulation. 
	We assume that the auditor has black-box access to $\cF$.
	In other words, the auditor can run $\cF$ on a set of inputs $\{ \fvec_j \}$ and observe its outputs $\{ Z_j  = \cF(\fvec_j) \}$.
	Note that the inputs $\fvec$ and $\fvec'$ need not correspond to real users and could represent hypothetical users.
\end{enumerate}

In this work,
we restrict our analysis to \textbf{counterfactual regulations}.
Specifically,
the auditor is given a regulation in the form:
	``The filtering algorithm $\cF$ should behave similarly under inputs $\fvec$ and $\fvec'$ for all $(\fvec, \fvec') \in \cS$.''
Below, we give two examples of counterfactual regulations.

\begin{example}\label{ex:covid}
	Suppose that the regulation prohibits targeted advertisements that are based on a user's indicated sexual orientation \citep{wachter2020affinity}.
	This can be written as ``the advertisements shown by $\cF$ should be similar when given two users that are identical except for their sexual orientations'',
	and $\cS$ could be a set of pairs $(\fvec, \fvec')$, where $\fvec'$ differs from $\fvec$ only in the  (hypothetical)  users' sexual orientation.\footnote{\label{ft:sexual_orientation}
		This example is simplified in order to illustrate the meaning behind a counterfactual regulation. 
		If one seeks to protect against more nuanced effects, such as proxy variables,  
		one could modify not only a user's sexual orientation but also any proxy variables. 
		Producing counterfactual inputs is out of the scope of this work. 
		We direct interested readers to texts on causal inference \citep{pearl2016causal} and on causality and fairness \citep{hu2020s,mishler2021fairness,kasirzadeh2021use}.
	}
	
\end{example}

\begin{example}
	Suppose that the regulation requires that articles containing medical advice on COVID-19 are robust to whether the user is left- or right-leaning. 
	This can be framed as ``the articles that are selected by $\cF$ and provide medical advice on COVID-19 should be similar for left- and right-leaning users'',
	and $\cS$ could be a randomly generated set of 
	left- and right-leaning user pairs.
\end{example}

The goal of this work is to enforce a regulation of the form ``$\cF$ should behave similarly under $\fvec$ and $\fvec'$  for all $(\fvec, \fvec') \in \cS$''.
The question remains: \emph{What is an appropriate notion of ``similarity''?}

\subsection{Decision robustness}\label{sec:decision_robustness}

We begin by observing that algorithmic filtering is powerful (and often harmful) because \emph{information influences decisions}: 
the content that a user sees can affect how they vote, 
whether they get vaccinated,
what restaurants they frequent, 
what items they purchase, 
and more. 
Stated differently,
if algorithmic filtering did not influence users' decisions, 
then there would be no desire to regulate it. 

Therefore,
if one seeks to enforce a counterfactual regulation,
the notion of ``similarity'' that is enforced should be with respect to the outcome of interest:
the users' decisions.
However, an auditor does not and should not have access to the users or their decisions (e.g., whether they get vaccinated). 
As such, the problem that the auditor faces can be stated as follows. 

Suppose that there are two identical (hypothetical) users. 
One is shown $\cF(\fvec)$ and the other is shown $\cF(\fvec')$. 
Suppose both users are given an identical set of queries $\cQ$ (e.g., where to eat dinner, whether to get vaccinated, what to wear). 
Let $\cD$ and $\cD'$, respectively, denote the (hypothetical) decisions that the first and second users make given queries $\cQ$.

Then, the auditor enforces similarity by ensuring \textbf{decision robustness} as follows:
\begin{quote}
	$\cF$ is \emph{decision-robust} to $(\fvec, \fvec')$ 
	if and only if, for any $\cQ$, 
      one cannot determine with high confidence that $\fvec \neq \fvec'$ from the decisions $\cD$ and $\cD'$.
\end{quote}
Decision robustness guarantees that the decision-making behavior of any user under $\cF(\fvec)$ and $\cF(\fvec')$ is indistinguishable with respect to $(\fvec, \fvec')$. 
However, ensuring decision robustness is challenging because the auditor does not have access to users or their decisions. 
Our objective is to \emph{provide an auditing procedure that guarantees decision robustness given only $\cS$ and black-box access to $\cF$}.

\subsection{Formalizing the auditor's goal} \label{sec:formal_goal_H_test}

Recall that $\cF$ is decision-robust to $(\fvec, \fvec')$---and therefore complies with the regulation---when, for any $\cQ$, one cannot determine with high confidence that $\fvec \neq \fvec'$ from $\cD$ and $\cD'$. 
In this section, we show that decision robustness can be expressed as a \textbf{binary hypothesis test}. 

Formally, 
consider a pair of inputs $(\fvec, \fvec') \in \cS$ and set of queries $\cQ$. 
To ``determine whether $\fvec \neq \fvec'$ from $\cD$ and $\cD'$'' is equivalent to using $\cD$ and $\cD'$ to decide between the following hypotheses:
\begin{align}
	H_0 : \genmodel( \fvec ) = \genmodel( \fvec' ) \qquad H_1 : \genmodel ( \fvec ) \neq \genmodel( \fvec' ) \label{eq:H_test}
\end{align}
To see this equivalence, observe that we can write the Markov chain $\bx \rightarrow \genmodel (\bx) \rightarrow Z \rightarrow \cD$. 
In other words, decisions $\cD$ depend on inputs $\bx$ only through the parameters $\genmodel(\bx)$. 
If one cannot determine that $\genmodel(\bx) \neq \genmodel (\bx')$ from $\cD$ and $\cD'$, 
then one also cannot determine that $\bx \neq \bx'$.

Let $H \in \{ H_0, H_1 \}$ denote the true (unknown) hypothesis.\footnote{Although the auditor has access to $S$ and knows whether $\fvec \neq \fvec'$,
	decision-robustness requires that one cannot determine this fact from $\cD$ and $\cD'$. Therefore, the hypothesis test treats $\fvec$ and $\fvec'$ as unknown.
}
Let $\hat{H} \in \{ H_0, H_1 \}$ denote the hypothesis that is chosen, where
the outcome $\hat{H} = H_1$ is equivalent to determining that $\genmodel (\fvec) \neq \genmodel( \fvec')$. 
We say that a test $\hat{H}$ is \textbf{$(1 - \epsilon)$-confident} that $\genmodel (\fvec) \neq \genmodel( \fvec')$
if $\hat{H} = H_1$ 
and $\bbP(\hat{H} = H_1 | H = H_0) \leq \epsilon \in [0, 1]$.\footnote{
	$\bbP$ is taken with respect to $p_{\bz}( \cdot \hspace{1pt} ; \genmodel (\fvec) )$ and $p_{\bz}( \cdot \hspace{1pt}  ; \genmodel (\fvec') )$.	
}
While one would like the test to be {$(1 - \epsilon)$-confident}, a {trivial} test that always chooses $H_0$ is $(1 - \epsilon)$ confident but has $100$ percent error when $H = H_1$. 
Therefore, one would also like the test to satisfy: $\bbP(\hat{H} = H_0 | H = H_1) \leq \alpha$ for some small $\alpha \in [0, 1]$. 
However, not all $\alpha \in [0,1]$
may be achievable while retaining the property of {$(1 - \epsilon)$-confident} for a given $\epsilon$.

To this end,
we turn to the \textbf{uniformly most powerful unbiased (UMPU)} test \cite{casella2021statistical,lehmann2005testing}. 
One can think of it as follows:
{if the UMPU test cannot determine that $\genmodel (\fvec) \neq \genmodel( \fvec')$ with high confidence,
then no other reasonable test can. }
Formally, 
suppose that one would like to find a test that maximizes the true positive rate (TPR) while ensuring  the false positive rate (FPR) is at most $\epsilon$ such that the test solves:
\begin{align}
	\max {\hspace{-1.5pt}}_{\hat{H}} \hspace{2pt}  \bbP(\hat{H} = H_1 | H = H_1 )  \qquad \text{s.t.}  \qquad  \bbP(\hat{H} = H_1 | H = H_0) \leq \epsilon , \label{eq:ump}
\end{align}
If a test $\hat{H}$ solves \eqref{eq:ump} for all $\genmodel(\fvec), \genmodel(\fvec')  \in \parFam$,
then it is the \textbf{uniformly most powerful (UMP)} test.
The UMPU test is the UMP test among all unbiased tests, where a test $\hat{H}$ is \textbf{unbiased} if:
\begin{align*}
	\bbP(\hat{H} = H_1 | H = H_1) \geq \delta \geq \bbP(\hat{H} = H_1 | H = H_0)
\end{align*}
for some $\delta \in [0, 1]$.\footnote{Intuitively, an unbiased test $\hat{H}$ ensures that the probability that $\hat{H}$ chooses $\fvec \neq \fvec'$  is always higher when $H_1 : \fvec \neq \fvec$ is true than when $H_0 : \fvec = \fvec'$ is true. }
Let	$\hat{H}_\epsilon^*$ denote the UMPU test, if it exists, and $\alpha^*(\epsilon) =  \bbP(\hat{H}_\epsilon^* = H_0 | H = H_1 )$. 
Intuitively, $\hat{H}_\epsilon^*$ is the test that is best at detecting when $\cF$ is not decision-robust (i.e., it maximizes the TPR) 
while making sure that it rarely falsely accuses $\cF$ of not being decision-robust (i.e., its FPR is at most $\epsilon$) among all reliable (i.e., unbiased) tests. 
Given the UMPU test, 
{decision-robustness can be formalized} as follows. For $\epsilon, \alpha \in [0, 1]$ and when the UMPU test exists,
\begin{center}
	{$\cF$ is ($\epsilon, \alpha$)-\emph{decision-robust} to $(\fvec, \fvec')$ $\iff$ for any $Q$, $\bbP( \hat{H}_\epsilon^* = H_0 | H =H_1) \leq \alpha $.}
\end{center}

\textbf{The goal of the auditing procedure}.
Therefore, determining whether a platform's filtering algorithm $\cF$ complies with a counterfactual regulation comes down to determining whether, for all  $Q$ and $(\fvec, \fvec') \in \cS$,
the UMPU test cannot confidently reject $H_0$ given $\cD$ and $\cD'$.
However, this task is not straightforward because the auditor's goal is to provide a guarantee on how $\cF$ affects users' decisions
\emph{without} access to the users or their decisions (i.e., without $Q$, $\cD$, or $\cD'$).
In this work, 
we show that it is possible to guarantee approximate asymptotic decision-robustness \emph{given only $\cS$ and black-box access to $\cF$} using insights from statistical learning and decision theory.

\section{Auditing procedure}\label{sec:audit}

In this section,
we present a procedure such that, 
\emph{given a regulation on algorithmic filtering that is expressed in counterfactual form, an auditor can test whether the platform's filtering algorithm is in compliance with the regulation}. 
In Section \ref{sec:theory},
we show that, if $\cF$ passes the audit, then $\cF$ is approximately asymptotically decision-robust.
In Section \ref{sec:cost_of_reg}, 
we study the cost of regulation and find that there are conditions under which the audit does not place a performance cost on the platform.

\subsection{Notation and definitions}

Before proceeding, we require some notation and definitions.  
 Recall that
 $\cF$ generates $Z$ by drawing $m$ samples from $p_{\bz}(\cdot \hspace{1pt} ; \genmodel(\fvec))$, 
 where $\genmodel(\fvec)$ is unknown. 
 In statistical inference \cite{lehmann2006theory},
 an \emph{estimator} is a mapping $\cL : \cZ \rightarrow \parFam$ such that $\cL(Z)$ is an estimate of the parameters $\genmodel$ that generated $Z$. 

 \begin{definition}
 	An estimator $\cL :  \cZ \rightarrow \parFam$ is \emph{unbiased} 
 	if and only if $\bbE_{p_{\bz}( \cdot \hspace{1pt} ; \genmodel ) }[  \cL ( Z ) ] = \genmodel$ for all $\genmodel \in \parFam$.
 \end{definition}
 
 \begin{definition}
 	When it exists, 
 	the \emph{minimum-variance unbiased estimator (MVUE)}  $\cL^+ :  \cZ \rightarrow \parFam$  is an estimator that is unbiased and has the lowest variance among all unbiased estimators, i.e., satisfies $\bbE_{p_{\bz}( \cdot \hspace{1pt} ; \genmodel ) } [ ( \cL^+ ( Z ) - \genmodel )^2 ]  \leq \bbE_{p_{\bz}( \cdot \hspace{1pt} ; \genmodel ) } [ ( \cL ( Z ) - \genmodel )^2 ]$ for all unbiased $\cL : \cZ \rightarrow \parFam$ and all $\genmodel \in \parFam$.
 \end{definition}
Let $\chi^2_r$ denote the chi-squared distribution with $r$ degrees of freedom
and $\chi^2_r ( q )$ be defined such that $\bbP( v \leq \chi^2_r ( q ) ) = q$ where $v \sim \chi^2_r$. 
Lastly, let $I(\genmodel) \in \mathbb{R}^{r \times r}$ denote the Fisher information matrix at $\genmodel$. 
An exact definition of $I(\genmodel)$ is given in the Appendix. 
Intuitively, $I(\genmodel)$  captures how well an estimator can learn $\genmodel$ from $Z$. 
As a simple example,
suppose $\bz_i$ are drawn i.i.d. from $\cN(\mu, \sigma^2)$, where $\sigma^2$ is known, $r = 1$, and $\genmodel = \mu$.
In general, it takes more samples to accurately estimate $\mu$ when the variance $\sigma^2$ is large, 
and, as expected, the Fisher information scales with $1/\sigma^2$.

\subsection{The audit}

In this section, we present the auditing procedure. 
Recall that a counterfactual regulation requires that $\cF$ behave similarly under $\fvec$ and $\fvec'$ for all $(\fvec, \fvec') \in \cS$.
Algorithm \ref{alg:test} provides a test for determining whether $\cF$ complies with the given regulation  
for a pair of inputs $(\fvec, \fvec') \in \cS$. 
To test other pairs in $\cS$,
simply repeat Algorithm \ref{alg:test} and modify $\bx$ and $\bx'$ accordingly. 
If $\hat{H}_\epsilon = H_1$ for any pair, then the platform does not pass the audit. 
Below, we list and explain several characteristics of the audit. 

\textbf{Scalability}.
Algorithm \ref{alg:test} is intentionally designed to be scalable. 
Scalability allows the auditor to construct the audit as they wish. 
For example, the auditor may wish to add more pairs to $\cS$ or to repeat the audit at different times.
Alternatively, the auditor may require not that
$\cF$ behave similarly under $\fvec$ and $\fvec'$ for \emph{all} $(\fvec, \fvec') \in \cS$ but for at least $(1 - \alpha) \in [0, 1]$ of them.\footnote{Here, $\alpha \in [0, 1]$ would correspond to the maximum allowable false negative rate (FNR).} 
To do so, the auditor can run Algorithm \ref{alg:test} over $\cS$ and, if the number of times
$\hat{H}_\epsilon = H_1$ exceeds $\alpha |\cS| $,
then the platform does not pass the audit.
Scalability also allows the auditor to see the pairs $(\fvec, \fvec')$ for which $\cF$ fails the test.

{
\begin{algorithm}[t]
	\SetAlgoLined
	\KwIn{Regulation parameter $\epsilon \in [0,1]$;
		model family $\parFam \subset \mathbb{R}^r$;
		black-box access to the filtering algorithm $\cF : \cX \rightarrow \cZ$;
		a pair of counterfactual inputs $(\fvec, \fvec') \in \cS$.}
	\KwResult{$\hat{H}_\epsilon = H_0$ if the test is passed; $\hat{H}_\epsilon = H_1$, otherwise.}
	
	\vspace{2pt}
	
	$\usermodel \leftarrow \cL^+ ( \cF ( \bx ) )$\; 
	$\usermodel{}' \leftarrow \cL^+ ( \cF ( \bx' ) )$\; 
	
	\If{$(\usermodel - \usermodel{}')^\top I ( \usermodel ) (\usermodel - \usermodel{}' ) \geq \frac{2}{m} \chi^2_r ( 1 - \epsilon ) $ \label{lin:test}}{
		Return $\hat{H}_\epsilon = H_1$\;
	}
	
	Return $\hat{H}_\epsilon = H_0$\; 
	\caption{Scalable version of auditing procedure}\label{alg:test}
\end{algorithm}

}

\textbf{Tunable parameter}. 
One benefit of the procedure is that the tunable parameter $\epsilon$ has an intuitive meaning. 
We see in Section \ref{sec:guarantee} that $\epsilon$ is a maximum FPR.
Capping the FPR ensures that the auditor is not distracted by red herrings
and prevents the auditor from investing the resources needed to investigate (or bring a case against) the platform unless they are at least $(1 - \epsilon)$-confident that the platform violates the regulation.
Decreasing $\epsilon \in [0,1]$ reduces the number of false positives
while increasing $\epsilon$ makes the regulation more strict (at the risk of receiving more false positives). 

\textbf{Advantages}. 
In addition to the benefits regarding scalability and $\epsilon$ discussed above, 
this procedure has two additional advantages. 
First, 
the procedure does not require access to users or their personal data. 
Second, it requires only black-box access to $\cF$, 
which means that an auditor does not need to know the inner-workings of $\cF$ (there is often resistance to giving auditors full access to $\cF$) 
and, perhaps more importantly, the procedure works even when $\cF$ changes internally.

\textbf{When the MVUE does not exist, use the MLE}. 
Recall that $\cL^+$ denotes the MVUE. 
We will see that there is a theoretical justification for using the MVUE (Proposition \ref{prop:mvue}). 
However, there are cases in which the MVUE does not exist but the maximum likelihood estimator (MLE) does. 
The MLE is a good substitute for the MVUE because the MVUE and MLE are often asymptotically equivalent \cite{portnoy1977asymptotic}.
When this asymptotic equivalence holds, 
using the MLE for $\cL^+$ gives the same theoretical guarantees (namely, Theorem \ref{thm:reg_hyp_test}) as the MVUE. 

\textbf{Symmetry}.
Algorithm \ref{alg:test} is not symmetric with respect to $\genmodel$ and $\genmodel'$ (or, equivalently, $\fvec$ and $\fvec'$). 
This can be useful if the auditor would like to have a baseline input $\fvec$ and run Algorithm \ref{alg:test} over different $\fvec'$. 
If the auditor would like symmetry, 
they may wish to run Algorithm \ref{alg:test} twice, swapping the order of  $\fvec$ and $\fvec'$, or to alter the Fisher information matrix in Line \ref{lin:test} to be $I((\tilde{\genmodel} + \tilde{\genmodel}{}')/2)$, if it exists.

\textbf{Choice of $\parFam$}.
Recall that $\parFam$ captures the set of possible generative models.
In choosing the model family $\parFam$, the auditor may find that a simple $\parFam$ is more tractable and interpretable while a complex $\parFam$ is more general.
As explained in Section \ref{sec:insight}, $\parFam$ can also be viewed as the set of possible cognitive models that users employ when making decisions. 
Therefore, the auditor may wish to choose $\parFam$ to be just rich enough to mirror the complexity of common cognitive models. 
\vspace{-1pt}

\section{Explaining the procedure and its theoretical guarantees}\label{sec:theory}

\vspace{-1pt}

Recall from Section \ref{sec:problem_statement} that a filtering algorithm $\cF$ complies with a counterfactual regulation if $\cF$ is decision-robust. 
In Section \ref{sec:guarantee}, we show that, if the platform passes the audit in Algorithm \ref{alg:test}, 
then $\cF$ is guaranteed to be  approximately asymptotically decision-robust. 
In Section \ref{sec:insight}, we provide insights on the role of the MVUE. 
All proofs are given in the Appendix.

\vspace{-1pt}
\subsection{Guarantee on the audit's effectiveness}\label{sec:guarantee}
\vspace{-1pt}

\begin{theorem} \label{thm:reg_hyp_test}
	Consider \eqref{eq:H_test}.
	Let $\genmodel^* = (\genmodel(\fvec)  + \genmodel (\fvec')) / 2$. 
	Suppose that $\bz_i$ and $\bz'_i$ are drawn i.i.d. from $p(\cdot \hspace{1pt} ; \genmodel (\fvec) )$ and $p(\cdot \hspace{1pt} ; \genmodel (\fvec') )$, respectively, for all $i \in [m]$
	and $\cP = \{ p_{\bz} ( \cdot \hspace{1pt} ; \genmodel ) : \genmodel \in \parFam \}$ is a regular exponential family that meets the regularity conditions stated in Appendix \ref{app:tech_details}.
	If $\hat{H}$ is defined as:
	\begin{align}
		\hat{H} = H_1 
		\iff
		(\cL^+(Z)  - \cL^+(Z'))^\top I ( \genmodel^* ) (\cL^+(Z)  - \cL^+(Z'))  \geq \frac{2}{m} \chi_r^2 ( 1 - \epsilon ) , \label{eq:thm_H1_def}
	\end{align}
	then $\bbP( \hat{H} = H_{1} | H = H_{0} ) \leq \epsilon$ as $m \rightarrow \infty$. 
	If $r = 1$, then {$\lim_{m \rightarrow \infty} \bbP( \hat{H} = H_{0} | H = H_{1} ) = \alpha^*(\epsilon)$.}	
\end{theorem}

\vspace{-1pt}
\textbf{Understanding the result}.
Recall that the goal of an auditor is to determine whether the platform's filtering algorithm $\cF$ is compliant with a given regulation by determining whether $\cF$ is decision-robust. 
Theorem \ref{thm:reg_hyp_test} confirms that the audit in Algorithm \ref{alg:test} enforces approximate asymptotic decision robustness. 
To see this connection, observe that the test in \eqref{eq:thm_H1_def} is identical to the test in Algorithm \ref{alg:test} 
with one substitution---$\cL^+ ( Z )$ is replaced by $\genmodel^*$---which 
implies that the test in \eqref{eq:thm_H1_def} is asymptotically equivalent to the audit.
Therefore, Theorem \ref{thm:reg_hyp_test} establishes that, if Algorithm \ref{alg:test} returns $\hat{H}_\epsilon = H_1$, 
then the auditor is $(1 - \epsilon)$-confident that $\cF$ is not decision-robust as $m \rightarrow \infty$.\footnote{
	We say that, if $\cF$ passes the audit, it is \emph{approximately} decision-robust because $\hat{H}_\epsilon$ is not the UMPU test, as defined in Section \ref{sec:problem_statement}. Obtaining a UMPU test is difficult for $r >1$, but the test $\hat{H}_\epsilon$ is not far from the UMPU test, as demonstrated by the fact that it is the UMPU test when $r = 1$. 
}

Intuitively,
if the platform passes the audit in Algorithm \ref{alg:test}, 
then the decision-making of any user that is shown $\cF(\bx)$ is guaranteed to be $\epsilon$-indistinguishable from their decision-making should they have been shown $\cF(\bx')$ instead.
Importantly, this guarantee is provided \emph{without} access to users or their decisions with the help of insights from statistical learning theory (see proof of Theorem \ref{thm:reg_hyp_test}).

\textbf{A few remarks}.
First, $\epsilon$ is a false positive rate (FPR). 
Decreasing $\epsilon$ increases the confidence that 
the auditor would like to have should it pursue action against the platform (see Section \ref{sec:audit}). 
Second, recall from Section \ref{sec:audit} that Algorithm \ref{alg:test} is modular and that, in most cases, 
one may wish to repeat it several times with different inputs. In that case, the result of Theorem \ref{thm:reg_hyp_test} holds for each individual run. 
Lastly, Theorem \ref{thm:reg_hyp_test} provides an asymptotic guarantee on the audit. 
In the next section, we show that, even for finite $m$, 
Algorithm \ref{alg:test} ensures a notion of decision robustness
by using the MVUE. 

\vspace{-1pt}
\subsection{Insight on the MVUE}\label{sec:insight}

Recall that the auditor's goal is to provide a guarantee with respect to users' decisions,
but the auditor does not have access to the users or their decisions. 
In this section, we provide  intuition for how the auditor enforces decision robustness without this information,
and we explain how the use of the MVUE in Algorithm \ref{alg:test} allows the auditor to enforce a notion of decision robustness for finite $m$. 
One can think of the MVUE as providing an ``upper bound'' on how much content $Z$ can influence a user's decisions.
This result is useful because it allows the auditor to reason about users' decisions without access to users or their decisions, both which can be expensive or unethical to obtain. 

\textbf{User model}. 
A user's decision-making process proceeds in three steps: 
the user observes information, updates their internal belief, then uses this belief to make a decision.
We model these steps as follows. 
Let $p_{\bz}(\cdot \hspace{1pt} ; \hat{\genmodel})$ denote the belief of a (hypothetical) user who is shown a feed $Z$, where $\hat{\genmodel} \in \parFam$.\footnote{The use of distributions to represent beliefs is common in the cognitive sciences \cite{chater2006probabilistic}. Note that, although representing beliefs as distributions is borrowed from Bayesian inference \cite{bernardo2009bayesian}, our representation does \emph{not} require that people are Bayesian (update their beliefs according to Bayes' rule), which is highly contested \cite{eberhardt2011confirmation}.}
Let the influence that $Z$ has on a user's belief be denoted by $\cL : \cZ \rightarrow \parFam$ 
such that $\hat{\genmodel} = \cL ( Z  )$.\footnote{A user's belief may be impacted by information other than $Z$ (e.g., the user's previous belief, news that they receive from friends, or content that they view on other platforms). This can be modeled by letting $\cL : \cZ \times \cJ \rightarrow \parFam$ where $J \in \cJ$ captures off-platform information. Because it does not change our results (see the Appendix), 
	we use $\cL : \cZ \rightarrow \parFam$ for notational simplicity.\label{ft:more_info_estimator} }
\begin{example}\label{ex:vaccine}
	As a highly simplified example, 
	suppose $\parFam = [-1, 1] \times [0, \infty)$ 
	and $p_{\bz}(\cdot ; \hat{\genmodel})$ is a Gaussian distribution with mean $\hat{\theta}_1$ and variance $\hat{\theta}_2$, 
	where $\hat{\theta}_1 = -1$ implies that the user does not believe that vaccines are effective, 
	$\hat{\theta}_1 = 1$ implies the opposite, 
	and $\hat{\theta}_2$ scales with the user's uncertainty in their belief. 
	If a user is easily influenced, 
	then 
	they might develop the belief $\cL(Z) = (-0.8, 0.1)$
	after being shown a feed $Z$ with anti-vaccine content.
	Alternatively, 
	the user could be confidently pro-vaccine 
	and very stubborn so that no matter what content they see,
	$\cL(\cdot) = (1, 0)$. 
\end{example}
Suppose the user is given a query $Q \in \cQ$, for which the user decides between two options: $A_0$ and $A_1$, e.g., whether or not to get vaccinated.
(Note that any decision between a finite number of options can be written as a series of binary decisions.)
Internally, the user places a value on each choice
such that, if the user were given $\genmodel$, the user would
choose $A_0$ if $v_0(\genmodel) \geq v_1(\genmodel)$ and $A_1$, otherwise. 

\begin{example}
	In microeconomic's utility theory, 
	$\genmodel$ would be an individual's preferences, 
	$v_i(\genmodel)$ would be the utility of choice $i$ under these preferences.
\end{example}
However, the user does not know $\genmodel$. 
They have a belief $\hat{\genmodel}$ that they use to infer whether $v_0( \genmodel ) \geq v_1(\genmodel)$.
In other words, the user's decision-making process is effectively a hypothesis test between: 
\begin{align}
	G_0 : v_0( \genmodel ) \geq v_1(\genmodel) \qquad G_1 : v_1 (\genmodel) > v_0 ( \genmodel ) . \label{eq:user_dm}
\end{align}
The following result motivates the use of the MVUE $\cL^+$ in the audit
by demonstrating that the MVUE enforces a finite-sample version of decision robustness. 

\begin{proposition}\label{prop:mvue}		
		Consider \eqref{eq:user_dm}. 
		Let $G\in \{ G_0 , G_1 \}$ denote the true (unknown) hypothesis.
		Suppose that $v_0 , v_1 : \parFam \rightarrow \bbR$ are affine mappings
		and there exists $\bu : \bbR^n \rightarrow \mathcal{U}$ such that, for $\cF(\bx) = \{ \bz_1, \hdots, \bz_m \}$, one can write $\bu(\bz_i) \sim p_{\bu}( \cdot ; \hspace{1pt} v_1(\genmodel(\bx)) - v_0( \genmodel(\bx) ))$ for all $i \in [m]$.
		Then, if the UMP test with a maximum FPR of $\rho$ exists, 
		it is given by the following decision rule: reject $G_0$ 
		(choose $A_1$) 
		when the minimum-variance unbiased estimate $\usermodel = \cL^+ (Z)$ satisfies $v_1( \usermodel )  - v_0 (\usermodel ) > \eta_\rho$ where $\bbP( v_1 ( \usermodel ) - v_0 ( \usermodel  ) > \eta_\rho | G= G_0 ) = \rho$;
		otherwise, accept $G_0$ (choose $A_0$). 
\end{proposition}

\textbf{Interpretation and implications}. 
 The auditor is interested in how users react to their content, which is captured by $\cL$. 
 However, $\cL$ may  difficult or even unethical to obtain.
For example, an auditor may wish to infer how advertisements affects a user's behavior, 
but doing so may require access to the user's personal data. 
Proposition \ref{prop:mvue} says that, under the stated conditions, if one wishes to study the impact of content on users' decisions, 
one can focus on 
the MVUE $\cL^+$
because, among all possible users, 
{the one whose decisions are most influenced by their content is the hypothetical user given by} the $\cL^+$. 
One can think of this hypothetical user as the ``most gullible user''.
Recall from Line \ref{lin:test} that Algorithm \ref{alg:test} requires $\cL^+(Z)$  and  $\cL^+(Z')$ to be sufficiently close in order for $\cF$ to pass the audit.
In combination, these observations demonstrate that the audit enforces a notion of decision robustness by ensuring that
the counterfactual beliefs $\cL^+(Z)$  and  $\cL^+(Z')$---and therefore the counterfactual decisions $\cD$ and $\cD'$---of the most gullible user are indistinguishable.

\textbf{Understanding the MVUE}. 
For intuition on why the MVUE is the ``most gullible user'',
recall that the MVUE is the unbiased estimator with the lowest variance. 
Suppose that a user's estimate  $\cL( Z)$ differs from $\cL^+( Z)$.
By definition, this estimate is biased or has higher variance.
When biased, the user's estimate is consistently pulled by some factor other than $Z$.
For example, a user who remains pro-vaccine no matter what content they see has a biased estimator. 
When the user's estimate has higher variance than the MVUE's,
it is an indication that the user places less confidence than the MVUE in what they glean from $Z$.
For example, the user could be skeptical of what they see on social media 
or scrolling very quickly and only reading headlines.
In this way, the MVUE corresponds to the user the user who ``hangs on every word''---whose decisions are most affected by their content $Z$. 

\section{Cost of regulation and the role of content diversity} \label{sec:cost_of_reg}

In this section, we turn our attention to how the auditing procedure affects
(a) the platform's ability to maximize an objective function $R$
and
(b) the type of content the platform is incentivized to filter when compliant with a regulation. 
In Section \ref{sec:cost}, 
we find that there are conditions under which the audit does not place a performance cost on the platform
and, intuitively, this occurs when the platform has enough degrees of freedom with which to filter.
We show in Section \ref{sec:diversity} that one of the ways the platform can increase $R$ while complying with the regulation is to add sufficient content diversity to users' feeds. 
Because diversity does not appear in the audit by design, 
this result suggests that content diversity can align the interests of regulators and platforms. 
All proofs are given in the Appendix.

\subsection{Cost of regulation}\label{sec:cost}

Suppose that the platform's goal is to maximize an objective function $R : \cZ \times \cX \rightarrow \bbR$---which we call \emph{reward}---while passing the audit.
For example, $R$ could be a measure of user engagement, 
user satisfaction, 
content novelty, 
or a combination of these and other factors. 
We leave $R$ unspecified, which means that our analysis holds for any choice of  $R$, unless otherwise stated.\footnote{There are settings in which $R$ is time-varying, e.g., when a platform's sources of revenue change with time. Making $R$ time-varying does not change our analysis or findings. 
	Therefore, we leave $R$ static for simplicity. \label{ft:static_R}
} 

Recall that $\cZ$ denotes the set of all possible feeds. 
Complying with a counterfactual regulation defined by $\cS$ is equivalent to restricting the platform's choice of feeds 
to a subset $\cZ(\cS) \subset \cZ$, which we call the \textbf{feasible set} under $\cS$.  
Intuitively, 
the stricter the regulation, the smaller the feasible set $\cZ(\cS) $.
If there is no regulation, then $\cZ(\cS) = \cZ$.
As such,
the platform's goal to maximize $R$ given inputs $\fvec$ while complying with the regulation can be expressed as:
\begin{align*}
	Z \in \arg \max_{W \in \cZ(\cS)} R(W, \fvec) . 
\end{align*}
Platforms are often interested in whether a regulation imposes a performance cost. 
To make this notion precise, we define the \textbf{cost of regulation} as follows. 
\begin{definition}
	The \emph{cost of regulation} for inputs $\bx$ is: $C =  \max_{W \in \cZ} R(W, \bx) - \max_{W \in \cZ(\cS)}  R(W, \bx)$. 
\end{definition}
A low cost of regulation implies that the platform can meet regulation without sacrificing much reward,
while a high cost of regulation implies that there is a strong performance-regulation trade-off.

\textbf{Performance-regulation trade-off}. 
Suppose that the feasible set shrinks from $\cZ$ to $\cZ(\cS)$.
Then, the maximum achievable reward
is affected in one of two ways. 
If the feasible set shrinks such that all reward-maximizing solutions in $\cZ$ are not contained in $\cZ(\cS)$, 
then the maximum reward \emph{decreases}, and the cost of regulation \emph{increases}. 
Alternatively, if the feasible set shrinks but at least one reward-maximizing solution in $\cZ$ is contained in $\cZ(\cS)$, then the maximum reward stays the same, and the cost of regulation \emph{does not increase}. 
Therefore, as long as $\cZ(\cS)$ preserves at least one (near) optimal solution, the cost of regulation is low.
The following result formalizes this notion. 
For this result, we overload the notation $R$ such that $R(\cL^+(Z) , \bx)$ denotes $R( Z, \bx)$.

\begin{theorem}\label{thm:cost_of_reg}
	Suppose there exists $\Omega \subset [r]$ where $1 < |\Omega| < r$
	such that
	$R(\genmodel_1, \bx) = R(\genmodel_2, \bx)$ if $\theta_{1,i} = \theta_{2,i}$ for all $i \notin \Omega$.
	Suppose that, for any $\genmodel \in \parFam$, $\beta > 0$ and $\bv \in \bbR^r$, 
	there exist a vector $\bar{\genmodel}_{\Omega}$ where $\bar{\theta}_{\Omega,i} = 0$ for all $i \notin \Omega$
	and a constant $\kappa > 0$ such that $\bv^\top I ( \genmodel + \kappa \bar{\genmodel}_{\Omega} ) \bv < \beta$ and $\genmodel + \kappa \bar{\genmodel}_{\Omega} \in \parFam$. 
	Then, if $m < \infty$, 
	there exists a set $\cZ$ such that the cost of regulation for $\bx$ under Algorithm \ref{alg:test} is $0$. 
\end{theorem}
\textbf{Interpretation of the result}.
Imposing a regulation restricts the platform's feasible set, which may place a performance cost on the platform.
However, a high cost of regulation is not inevitable. Indeed, if the feasible set contains at least one (near) optimal solution, then the cost of regulation is low. 
Theorem \ref{thm:cost_of_reg} provides a set of conditions under which the cost of regulation is low. 
Intuitively, the result states that, when $R$ is independent of at least one element in the parameter vector $\genmodel$ and that element has sufficient leverage over the Fisher information, 
then as long as the amount of content in a given feed is finite (i.e., $m < \infty$)
and the available content $\cZ$ is expressive enough, 
then the platform can always construct a feed from $\cZ$ that passes the audit without sacrificing reward. 

One may ask whether the conditions in Theorem \ref{thm:cost_of_reg} are feasible. 
To illustrate that the conditions are achievable, consider the following highly simplified example. 
\begin{example}\label{ex:gaussian}
	Suppose $\genmodel = (\mu , \sigma^2)$, $\parFam = \bbR \times \bbR_{\geq 0}$, 
	and $\cP = \{ \mathcal{N}(\mu , \sigma^2) : (\mu , \sigma^2) \in \parFam \}$. 
	In other words, we consider the family of 1-D Gaussian distributions. 
	If $R$ is a function of the mean $\mu$ but not the variance $\sigma^2$, 
	then Theorem \ref{thm:cost_of_reg} applies. 
	To see this, observe that $\Omega = \{ 2 \}$ and
	 let $\bar{\genmodel}_{\Omega} = (0 , 1)$.
	 Recalling that $I	( (\mu, \sigma^2) ) = \emph{diag}( \sigma^{-2} , \sigma^{-4}  / 2)$, 
	 the entries of $I(\genmodel)$ can be made arbitrarily small by increasing $\theta_2 = \sigma^2$. 
	 By Line \ref{lin:test} in Algorithm \ref{alg:test}, 
	 the smaller the entries of $I(\genmodel)$, the easier it is for the platform to pass the audit. 
	 Therefore, if the platform  has a high-reward feed $Z^*$ that is not in the feasible set, 
	 the platform can still achieve $R(Z^*, \bx)$ by increasing the feed's variance to obtain a new feed $Z$ that is in the feasible set. 
	 As long as $Z$ and $Z^*$ share $\mu$, $R(Z , \bx) = R(Z^*, \bx)$.
\end{example}

Theorem \ref{thm:cost_of_reg} provides conditions under which there is \emph{no cost of regulation}. 
These conditions can be relaxed if we are interested in scenarios for which the cost of regulation is low but not zero. 
We build this intuition in the next section, 
studying one of the ways a platform can achieve high reward while remaining compliant with regulation.

\subsection{Content diversity}\label{sec:diversity}

In this section, 
we show that one way that the platform can increases its reward while complying with the regulation is to 
ensure that the feeds have a sufficient amount of content diversity. 
These results suggest that \emph{content diversity may help to align the interests of regulators and platforms}.

We first formalize content diversity, then show how it relates to passing the proposed audit. 
\begin{definition}\label{def:diversity}
	For $\bv \in \bbR^r$ and $Z_0, Z_1 \in \cZ$, 
	feed $Z_0$ has greater \emph{content diversity} than $Z_1$ along $\bv$ if the Fisher information matrices at $\cL^+(Z_0)$ and $\cL^+(Z_1)$ satisfy:
	$\bv^\top (  I (\cL^+(Z_1)) - I (\cL^+(Z_0) ) \bv > 0$. 
\end{definition}

\textbf{Interpretation}. 
This definition says that the ``smaller'' the Fisher information,
the higher the content diversity. 
The Fisher information matrix $I(\genmodel)$ can be viewed as 
a measure of how easy it is to learn $\genmodel$ from a feed $Z$ that is generated by $p_{\bz}(\cdot \hspace{1pt} ; \genmodel)$. 
Consequently,
when $Z$ and $Z'$ have low content diversity along $(\genmodel(\bx) - \genmodel(\bx'))$,
an auditor can, without much effort,
learn that $\genmodel(\bx)$ is different from $\genmodel(\bx')$
and therefore that $\bx \neq \bx'$. 
Recall from Section \ref{sec:decision_robustness} that being able to say with high confidence that $\bx \neq \bx'$ implies that $\cF$ is not decision-robust and therefore does not comply with regulation. 
In this way, low content diversity reduces the likelihood that $\cF$ passes the audit.\footnote{
	Content diversity can also be understood in the context of Section \ref{sec:insight}. 
	Suppose that the content diversities of $Z$ and $Z'$ are very low. For example, suppose that $Z$ contains only pro-vaccine content and $Z'$ contains only anti-vaccine content.
	Then, the MVUE would learn a strong relationship between vaccines with positive outcomes from $Z$ and vice versa for $Z'$. If $\cQ$ contains a query about vaccines, $\cD$ and $\cD'$ would reflect these strong beliefs. In this way, $\cF$ is less likely to be decision-robust when the content diversity is low. On the other hand, if the content for all users contains both pro- and anti-vaccine content, then $\cD$ and $\cD'$ are more similar. 
}
Note that the Fisher information matrix captures two notions of diversity---the diversity of topics in a feed and the diversity of perspectives on each given topic in the feed---simultaneously.

\textbf{Connection between content diversity and the cost of regulation}.
Recall from Algorithm \ref{alg:test} that $\cF$ passes the audit when $ (\usermodel - \tilde{\genmodel}{}')^\top I ( \usermodel ) (\usermodel - \tilde{\genmodel} {}')$ is below some threshold. 
The platform can therefore pass the audit by ensuring that $I(\usermodel)$ is sufficiently ``small''.
By Definition \ref{def:diversity}, 
whether the Fisher information is ``small'' is precisely an indication of how much diversity is in a feed. 

In this way, increasing content diversity gives the platform more leeway when filtering. 
By ``shrinking'' the Fisher information,
the platform obtains more flexibility in setting $(\usermodel - \usermodel{}')$.
Stated differently, if the Fisher information $I(\usermodel)$ is ``large'', 
then the platform is more constrained because $(\usermodel - \usermodel{}')$ must be very small in order for the platform to pass the audit. 
Therefore, if the platform has a high-reward feed that does not pass the audit (i.e., is not in the feasible set), 
then the platform can generally maintain a high reward while complying with the regulation by adding content diversity.\footnote{
	The platform does not increase content diversity indefinitely because the platform 
	must also ensure that other terms in Line \ref{lin:test}---specifically, $(\usermodel - \usermodel{}')$---do not cause the platform to fail the audit. 
} 

Because content diversity is not part of the audit by design, this result is unexpected.
It states that the audit naturally {incentivizes} the platform to include a sufficient amount of content diversity with respect to $\fvec - \fvec'$. 
Returning to Example \ref{ex:covid}, if regulators require that medical advice on COVID-19 be robust to whether  a user is left- or right-leaning, 
then the differences between the medical advice shown to users across the political spectrum is captured by  
$\fvec - \fvec'$. 
Adding content diversity along this dimension means that right-leaning users receive medical advice on COVID-19 not only from right-leaning news outlets, but also from left-leaning ones, and vice versa.

\section{Background \& related work}\label{sec:related_work}

Algorithmic filtering \cite{devito2017algorithms,bozdag2013bias} has the potential to greatly improve  the user experience, but it can also yield unwanted side effects, like the spread of fake news \cite{marwick2017media,chesney2019deep,dcms2019disinformation}, over-representation of polarizing opinions  due to comment ranking \cite{siersdorfer2014analyzing}, amplification of echo chambers due to filter bubbles \cite{flaxman2016filter,pariser2011filter}, or advertising of products based on discriminatory judgments about user interests
\cite{speicher2018potential,sweeney2013discrimination,kim2018discrimination}. 
Although the severity of these effects is contested---for example, some studies argue that political polarization and echo chambers are not always products of internet use or data-driven algorithms \cite{boxell2017greater,hosseinmardi2020evaluating}---social media platforms are under rising scrutiny.

In response,
some platforms have begun to self-regulate \cite{medzini2021enhanced,democracy2018canada}. 
For example, 
Facebook has established an internal ``Supreme Court'' that reviews the company's decisions \cite{facebookOversightBoard}, 
Twitter has banned accounts associated with the spread of conspiracy theories \cite{twitter2021qanon}, 
YouTube has removed videos it views as encouraging violence \cite{youtube2021violence}, 
and so on. 
Self-regulatory practices offer various benefits,
such as the ability to adapt quickly to a changing social media ecosystem 
and 
the comparatively greater information access afforded to internal auditors than external ones.
However, many argue that self-regulation is insufficient and that governmental regulations are necessary in order to ensure that audits are executed by independent bodies.
Several regulations exist, 
such as the EU's
General Data Protection Regulation \cite{gdpr}
and Germany's Network Enforcement Act \cite{bundestag2017act}. 
There are also ongoing efforts, 
such as the push to review Section 230 of the U.S. Communications Decency Act \cite{cda}. 

Designing such regulations and auditing procedures remains a challenging problem, in part due to the number of stakeholders. 
For one, 
there are various legal and social obstacles facing regulations \cite{klonick2017new,brannon2019free,berghel2017lies,obar2015social},
including concerns that regulations might damage free speech or public discourse; violate personal rights or privacy;  transfer agency away from users to technology companies or governmental bodies; draw subjective lines between acceptable and unacceptable behavior;
or set precedents that are difficult to reverse. 
In light of the thriving exchange of goods between users, platforms, advertisers  and influencers that is facilitated by social media, 
many also fear that regulations may hurt innovation, 
lead to worse personalization, 
or block revenue sources \citep{evans2020reg}.

Current efforts to regulate content moderation  generally focus on specific issues, such as whether content is inappropriate (e.g., posts that contain hate speech  or bullying \cite{bbcReg,davidson2017automated}); discriminatory (e.g., race-based advertising \cite{angwin2017housing,speicher2018potential,sweeney2013discrimination,kim2018discrimination}); 
divisive (e.g., comment ranking algorithms that favor polarizing comments \cite{siersdorfer2014analyzing}); insulating (e.g., filter bubbles \cite{flaxman2016filter}), or misleading (e.g., fake news \cite{marwick2017media,chesney2019deep,dcms2019disinformation}). 
These works generally use one of the following strategies:  increasing content diversity (e.g., 
adding heterogeneity to recommendations \citep{bozdag2015breaking,helberger2018exposure}); drawing a line in the sand (e.g., determining whether discrimination has occurred by thresholding the difference between two proportions \citep{chouldechova2017fair}); or finding the origin of the content (e.g., reducing fake news by whitelisting news sources \citep{berghel2017lies}). 
In this work, 
we provide a general procedure such that, given a regulation in counterfactual form, 
an auditor can test whether the platform's filtering algorithm is compliant.
Our aim is to audit with respect to the outcome of interest in order to avoid unwanted side effects. 
Of particular note is that the proposed procedure does not require access to users  or their personal data.

We also consider how an audit affects the platform's ability to maximize an objective function as well as the content that the platform is incentivized to filter for a user. 
Our formulation is an instance of constrained optimization and bears resemblance to robust optimization \cite{bertsimas2011theory,wiesemann2014distributionally}. 
For instance, our definition of the cost of regulation mirrors the ``price'' of robustness studied in other works \cite{bertsimas2004price,bertsimas2011price}. 
Similarly, the performance-regulation trade-off that we discuss echoes the trade-offs that appears in other problems in which there are fairness \cite{bertsimas2011price,feldman2015certifying,kamishima2011fairness} and privacy \cite{campbell2015privacy,hirsch2010law} constraints. 
Our findings that there are conditions under which content diversity aligns the interests of regulators and platforms adds to the conversation on presenting different viewpoints on social media.
For instance, Levy  \cite{levy2021social} finds that users respond well when presented with multiple political viewpoints---which Levy terms counter-attitudinal content---even ones from an opposing political party. 
Increasing content diversity is also at the heart of other methods, including those for bursting filter bubbles \cite{bozdag2015breaking} or improving comment ranking \cite{giannopoulos2015algorithms}.
It is worth a note that our definition of content diversity captures two notions: 
diversity in the topics as well as diversity in the viewpoints on each topic. 

As a final remark, our analysis has parallels with differential privacy \cite{dwork2006calibrating,dwork2014algorithmic} in that it compares outcomes under different interventions 
\cite{wasserman2010statistical}. 
However it differs in the techniques required to study similarity. 
Our work also touches on aspects of but is distinct from social learning and opinion dynamics \cite{sl1, sl2, sl3, sl4} in that we study how information affects the beliefs of individuals.



\newpage

\ack

We would like to thank our anonymous reviewers and area chair for their time and suggestions. 
We would also like to thank the many people who provided feedback along the way, 
including but certainly not limited to Martin Copenhaver, 
Hussein Mozannar, 
Aspen Hopkins, 
Divya Shanmugam, 
and Zachary Schiffer. 
This work was supported in parts by the MIT-IBM project on "Representation Learning as a Tool for Causal Discovery", 
the NSF TRIPODS Phase II grant towards Foundations of Data Science Institute, the
Chyn Duog Shiah Memorial Fellowship, and the Hugh Hampton Young Memorial Fund Fellowship.\\[-4pt]

\emph{This work has been accepted to the 35th Conference on Neural Information Processing Systems (NeurIPS 2021), Sydney, Australia.}\\[-10pt]

\bibliography{ref.bib}{}
\bibliographystyle{plain}


\newpage
\section*{Checklist}

\begin{enumerate}
	
	\item For all authors...
	\begin{enumerate}
		\item Do the main claims made in the abstract and introduction accurately reflect the paper's contributions and scope?
		\answerYes
		\item Did you describe the limitations of your work?
		\answerYes We describe the scope of our work in Section \ref{sec:related_work}. We also discuss the limitations of our work throughout Sections \ref{sec:problem_statement}-\ref{sec:cost_of_reg}. We make a few additional remarks in Appendix \ref{app:discussion}.
		\item Did you discuss any potential negative societal impacts of your work? \answerYes We discuss the potential negative impacts of the work in Appendix \ref{app:discussion} and in Section \ref{sec:cost_of_reg}.
		\item Have you read the ethics review guidelines and ensured that your paper conforms to them? \answerYes
	\end{enumerate}
	
	\item If you are including theoretical results...
	\begin{enumerate}
		\item Did you state the full set of assumptions of all theoretical results?
		\answerYes
		\item Did you include complete proofs of all theoretical results?
		\answerYes All proofs are given in the Appendix. 
	\end{enumerate}
	
	\item If you ran experiments...
	\begin{enumerate}
		\item Did you include the code, data, and instructions needed to reproduce the main experimental results (either in the supplemental material or as a URL)?
		\answerNA We did not run experiments. 
		\item Did you specify all the training details (e.g., data splits, hyperparameters, how they were chosen)?
		\answerNA
		\item Did you report error bars (e.g., with respect to the random seed after running experiments multiple times)?
		\answerNA
		\item Did you include the total amount of compute and the type of resources used (e.g., type of GPUs, internal cluster, or cloud provider)?
		\answerNA
	\end{enumerate}
	
	\item If you are using existing assets (e.g., code, data, models) or curating/releasing new assets...
	\begin{enumerate}
		\item If your work uses existing assets, did you cite the creators?
		\answerNA We do not use existing assets or curate/release assets. 
		\item Did you mention the license of the assets?
		\answerNA
		\item Did you include any new assets either in the supplemental material or as a URL?
		\answerNA
		\item Did you discuss whether and how consent was obtained from people whose data you're using/curating?
		\answerNA
		\item Did you discuss whether the data you are using/curating contains personally identifiable information or offensive content?
		\answerNA
	\end{enumerate}
	
	\item If you used crowdsourcing or conducted research with human subjects...
	\begin{enumerate}
		\item Did you include the full text of instructions given to participants and screenshots, if applicable?
		\answerNA We did not use crowdsourcing or conduct research with human subjects. 
		\item Did you describe any potential participant risks, with links to Institutional Review Board (IRB) approvals, if applicable?
		\answerNA
		\item Did you include the estimated hourly wage paid to participants and the total amount spent on participant compensation?
		\answerNA
	\end{enumerate}
	
\end{enumerate}

\newpage
\appendix

\begin{center}
	\textbf{\LARGE Appendix}
	\\[16pt]
\end{center}

\setcounter{theorem}{0}
\section{Toy example}\label{app:toy}

In this section, we provide and expand upon a toy example. 
Recall that the inputs $\fvec$ and $\fvec'$ need not correspond to real users 
but could instead represent hypothetical users.

\begin{example}
	Suppose that the regulatory guideline requires that users in the same geographical location receive similar weather forecasts. 
	This can be written as ``the weather forecasts that are selected by $\cF$ should be similar for all users in the same geographical location'',
	and $\cS$ could be a randomly generated set of user pairs, where each pair corresponds to two (hypothetical) users in the same geographical location, and $\cS$ could contain pairs across many locations. \label{ex:weather}
\end{example}

\begin{figure}[h!]
	\centering
	\includegraphics[width=0.85\textwidth]{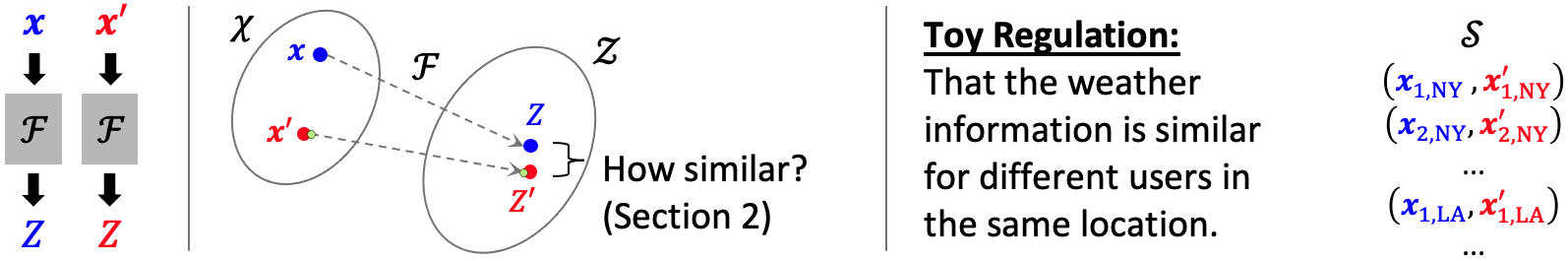}
	\caption{Visualization of toy example.}\label{fig:CF}
\end{figure}

Figure  \ref{fig:CF} visualizes counterfactual regulations. 
In the left-most panel, a filtering algorithm $\cF$ takes in counterfactual inputs $\bx$ and $\bx'$ and produces the feeds $Z$ and $Z'$. 
The middle panel visualizes this relationship graphically. 
Because a counterfactual regulation requires that $\cF$ behave similarly under $\bx$ and $\bx'$, 
the regulation is effectively requiring that the feeds $Z$ and $Z'$ are sufficiently similar (or, graphically, that they are close in $\cZ$).
The question of how to quantify ``similarity'' is addressed in Section \ref{sec:decision_robustness}.
The toy example in Example \ref{ex:weather} is illustrated in the right-most panel. 
Requiring that the weather information is similar for users in the same location can be tested by randomly selecting pairs of users $(\bx, \bx')$ in the same location,
placing these pairs in $\cS$, 
then running the audit over $\cS$. \\

\begin{figure}[h]
	\includegraphics[width=\textwidth]{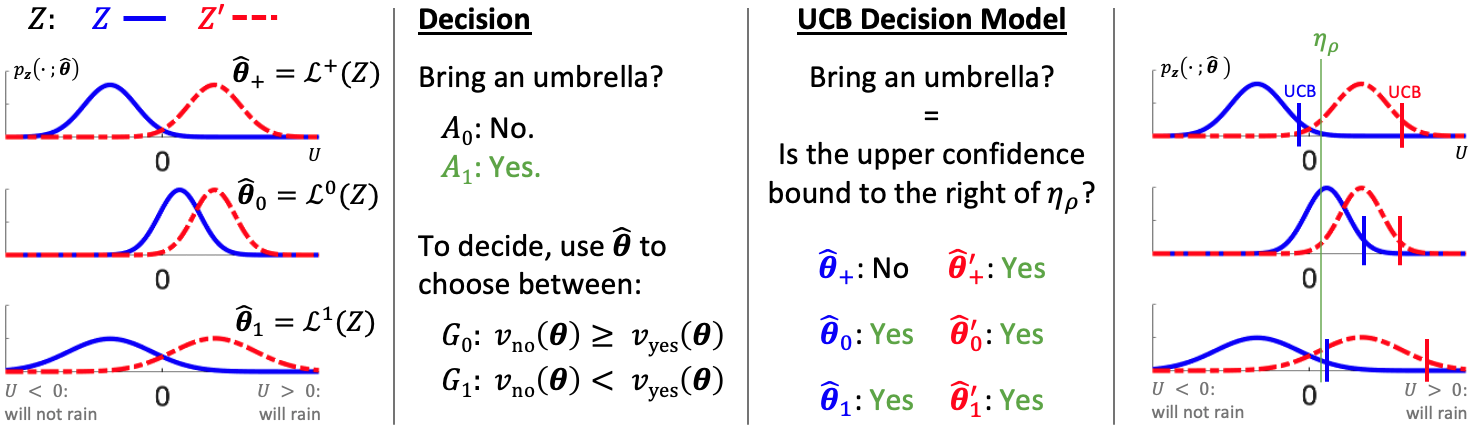} 
	\caption{
		Understanding the role of the MVUE (see Section \ref{sec:insight}). 
	} \label{fig:usermodelfig}
\end{figure}

Figure \ref{fig:usermodelfig} visualizes the intuition behind the MVUE discussed in Section \ref{sec:insight}. 
Specifically, it illustrates why the MVUE $\cL^+$ corresponds to the user whose decisions are most sensitive to $Z$. 
Suppose that a user's feed $Z$ contains content about the chance of rain and the user is deciding whether to bring an umbrella. 
Suppose the content in $Z = \cF(\fvec)$ reflects the actual chance of rain while $Z' = \cF(\fvec')$ contains disproportionately more content suggesting that it will rain. 
Perhaps $Z$ is shown to children while $Z'$ is shown to adults to encourage them to buy umbrellas.
Let there be three hypothetical users with estimators $\cL^+$, $\cL^1$, and $\cL^2$, as indicated in the left-most panel. (As discussed in Section \ref{sec:insight}, every user ingests their content differently. An estimator $\cL$ is simply a mapping from a feed $Z$ to the user's belief. In this toy example, we study three hypothetical users.)

In the left-most panel, 
each plot visualizes one of the user's belief $p_{\bz}(\cdot ; \hat{\genmodel})$ about whether it will rain $U$ given $Z$ (in solid blue) or given $Z'$ (in dashed red), 
where $U < 0$ suggests that it will not rain and
$U > 0$ suggests that it will rain.
$\cL^+$ is the MVUE, $\cL^1$ is a biased estimator (it is biased to the right such that the user tends to believe it will rain today no matter what the forecasts say), and $\cL^2$ is an unbiased estimator with higher variance than $\cL^+$ (the user does not put much confidence in the forecasts, so its belief is less ``peaky'' than the MVUE's).

In the second  panel, we write the decision of whether to bring an umbrella in terms of the setup in Section \ref{sec:insight}. Specifically, if the user knew that the true chance of rain as given by $\genmodel$, they would bring an umbrella if $v_{\text{yes}} (\genmodel) > v_{\text{no}}(\genmodel)$ and would not bring an umbrella, otherwise, 
where $v_i$ denotes the value that the user places on each option.
For example, $v_i$ may balance the user's dislike of carrying an umbrella with the user's  dislike of walking in the rain, and $v_i$ may differ across individuals. 

The third panel explains how the user would make a decision under the upper confidence bound (UCB) decision model, a popular model in the bandit literature \cite{bubeck2012regret}. 
Here $\genmodel$ captures the reward and sampling history of the bandit (i.e., the past experiences of a user with respect to rain and weather forecasts),
and $v_i (\genmodel)$ would give the UCB of arm $i$ (i.e., of the choices to and not to bring an umbrella).
As written in the third panel, 
under the UCB decision model, 
the user would choose to bring an umbrella if the UCB of their belief is to the right of some threshold $\eta_\rho$ (for details on $\eta_\rho$, see Section \ref{sec:insight})
and would not bring an umbrella, otherwise. 

In order to understand what decision each of the three users corresponding to $\cL^+$, $\cL^1$, and $\cL^2$ would do, 
examine the fourth  (right-most) panel. 
Let the threshold $\eta_\rho$ be given by the thin vertical line, as marked. 
Let the UCB for $Z$ and $Z'$ be given by the blue and red thick lines, as indicated in the top-most  plot (the blue line is always to the left of the red line). 
We see that, for this choice of $\eta_\rho$, 
the MVUE would not choose to bring an umbrella under $Z$ but would choose to do so under $Z'$. 
We also see that the users corresponding to $\cL^1$ and $\cL^2$ would choose to bring umbrellas under both $Z$ and $Z'$. 
These choices are also written in the third-panel from the left.

The goal of Figure \ref{fig:usermodelfig} is to provide intuition for why the MVUE corresponds to the ``most gullible user'': the hypothetical user whose decisions are most affected by their content.
Recall that $Z'$ indicates that it is more likely to rain than $Z$. 
As illustrated in the example, the MVUE is the only estimator among the three for which the user's decision is different when shown $Z$ versus $Z'$, 
whereas the users corresponding to the other estimators are less affected by the content that they see: their decisions remain the same under $Z$ and $Z'$. 
This example confirms the discussion in Section \ref{sec:insight} that the decisions of the MVUE are more sensitive to whether the content is $Z$ or $Z'$ than the decisions of other users (i.e., other estimators). 
Therefore, if we wish to enforce similarity between users' decision-making behavior under $Z$ and $Z'$---or, equivalently, under the inputs $\bx$ and $\bx'$---then the MVUE provides an ``upper bound'' on the sensitivity of users' decisions to their content. 

For intuition on why the MVUE is the ``most gullible user'',
recall that the MVUE is the unbiased estimator with the lowest variance. 
Suppose that a user's estimate  $\cL( Z)$ differs from $\cL^+( Z)$.
By definition, this estimate is biased or has higher variance.
When biased, the user's estimate is consistently pulled by some factor other than $Z$.
For example, a user who remains pro-vaccine no matter what content they see has a biased estimator. 
When the user's estimate has higher variance than the MVUE's,
it is an indication that the user places less confidence than the MVUE in what they glean from $Z$.
For example, the user could be skeptical of what they see on social media 
or scrolling very quickly and only reading headlines.
In this way, the MVUE corresponds to the user the user who ``hangs on every word''---whose decisions are most affected by their content $Z$.\\

\begin{figure}[h]
	\centering
		\includegraphics[height=4.45cm]{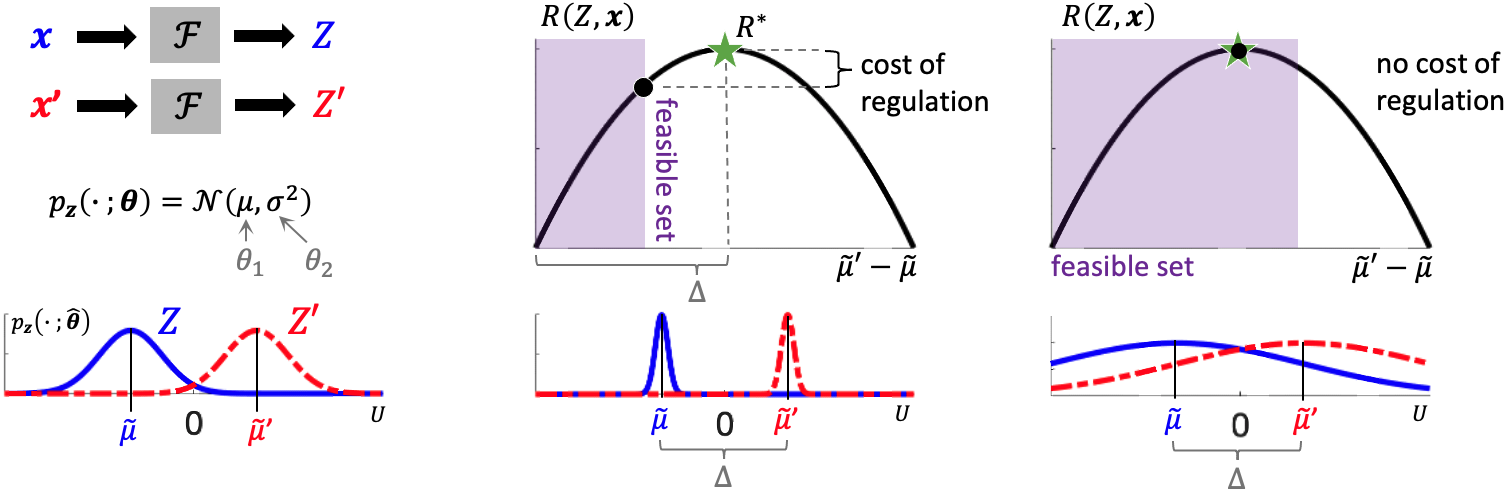}
		\caption{Visualizing the cost of regulation and the connection to content diversity.} \label{fig:cost_of_reg}
\end{figure}

Figure \ref{fig:cost_of_reg} visualizes the cost of regulation and illustrates why increasing content diversity can reduce the cost of regulation. 
In this example, suppose that $p_{\bz}(\cdot ; \genmodel)$ are 1-D Gaussian, as in Example \ref{ex:gaussian}. 
We examine how different choices of $(Z, Z')$ affect (a) the feasible set and (b) the platform's cost of regulation. 
In the left-most panel, we re-iterate that $Z = \cF(\bx)$ and $Z' = \cF(\bx')$. 
As stated above, we assume that $\parFam$ parameterizes the family of 1-D Gaussian. 
The bottom of the left-most panel provides an example of how an estimator $\cL$ would behave if given two feeds $Z$ and $Z'$ where $Z$ is a 1-D Gaussian centered to the left of $0$, $Z'$ is a 1-D Gaussian centered to the right of $0$, and both have the same variance $\sigma^2$. 
For this example, we assume that the variance $\sigma^2$ is the same for $Z$ and $Z'$.
Specifically, the distribution under $\hat{\genmodel} = \cL(Z)$ is plotted in solid blue while the distribution under $\hat{\genmodel} = \cL(Z')$ is plotted in dashed red. For the purposes of this example, one can assume that $\cL = \cL^+$ is the MVUE, and we denote the estimate of $\genmodel$ by $\tilde{\genmodel}$ to be consistent with the notation in the main text. 

The middle and right-most panels visualize the cost of regulation under two different choices of $(Z, Z')$. 
In the middle panel, 
$Z$ and $Z'$ are chosen such that the distributions $p(\cdot \hspace{1pt} ; \tilde{\genmodel})$ under $Z$ and $Z'$ are given in the middle-bottom. 
Specifically, $\tilde{\mu}$ and $\tilde{\mu}'$ are $\Delta$ apart, and
the variance $\sigma^ 2$ is fairly small. 

Suppose that, as in Example \ref{ex:gaussian}, 
$R$ is a function of the means (i.e., of $\tilde{\mu}$ and $\tilde{\mu}'$) but not of the variance $\sigma^2$. 
Furthermore, suppose that the platform maximizes its reward $R$ when $\tilde{\mu}' - \tilde{\mu} = \Delta$, as visualized by the green star in the top middle plot. 
As explained in Section \ref{sec:cost_of_reg}, 
a regulation restricts the platform's choice of feeds from $\cZ$ to the feasible set $\cZ(\cS)$, which is a  subset of $\cZ$. 
In the top middle panel, we visualize the feasible set in purple. 
Specifically, for the given choice of $\sigma^2$ (as plotted in the bottom middle), the feasible set does not include $\tilde{\mu}' - \tilde{\mu} = \Delta$. The maximum reward that the platform can achieve under the regulation for this choice of $\sigma^2$ is indicated by the black dot, 
and the vertical distance between the green star and black dot is the cost of regulation. 

However, if the platform still chooses $Z$ and $Z'$ such that $\tilde{\mu}' - \tilde{\mu} = \Delta$ but increase the variance $\sigma^2$, the story is different. 
In the right-most panel, we show that, for a larger $\sigma^2$, 
the feasible set expands to include $\tilde{\mu}' - \tilde{\mu} = \Delta$. 
As such, the reward-maximizing solution is contained within the feasible set, and there is no cost of regulation. 
In this way, adding a sufficient amount of content diversity can reduce the cost of regulation, thereby allowing the platform to achieve high reward while also complying with the regulation. 

Mathematically, this example is explained by Theorem \ref{thm:cost_of_reg} and discussed in Example \ref{ex:gaussian}. 
Expanding on the discussion in Example \ref{ex:gaussian}, 
recall that, when  $\genmodel = (\mu , \sigma^2)$, $\parFam = \bbR \times \bbR_{\geq 0}$,
$\cP = \{ \mathcal{N}(\mu , \sigma^2) : (\mu , \sigma^2) \in \parFam \}$ is the family of 1-D Gaussian distributions, 
and $R$ is a function of $\mu$ but not $\sigma^2$ (as in Figure \ref{fig:cost_of_reg}), 
then $I	( (\mu, \sigma^2) ) = \text{diag}( \sigma^{-2} , \sigma^{-4}  / 2)$, 
$\bar{\genmodel}_{\Omega} = (0, 1)$,  
and
\begin{align}
	\bv^\top I(\genmodel + \kappa \bar{\genmodel}_\Omega ) \bv = v_1^2 / (\sigma^2 + \kappa) + v_2^2 / (2 ( \sigma^2 + \kappa)^2) \label{eq:thresh_gauss_ex} .
\end{align} 
This quantity becomes very small if $\kappa$ is very large. 
It turns out that making this quantity small is precisely what we want and that
making $\kappa$ large is the same as increasing the content diversity of a feed because $\genmodel + \kappa \bar{\genmodel}_\Omega = (\mu, \sigma^2 + \kappa)$. 
To see this connection, recall that a pair of feeds $Z$ and $Z'$ passes regulation if 
$(\cL^+(Z)  - \cL^+(Z'))^\top I (\cL^+(Z)  ) (\cL^+(Z)  - \cL^+(Z'))  < \frac{2}{m} \chi_r^2 ( 1 - \epsilon ) $, 
Therefore, given a feed $Z^*$ for which $\cL^+(Z^*) = \tilde{\genmodel}{}^*$ and $R(Z^*, \bx)$ is the maximum achievable reward, 
one can create a new feed $Z$ that passes regulation by taking $Z^*$ and increasing its content diversity such that $\cL^+(Z) = (\mu^*, (\sigma^*)^2 + \kappa)$. 
By \eqref{eq:thresh_gauss_ex}, the quantity $(\cL^+(Z)  - \cL^+(Z'))^\top I (\cL^+(Z)  ) (\cL^+(Z)  - \cL^+(Z'))$  can be made arbitrarily small by taking $\kappa$ to be large, which means that $Z$ is in the feasible set. 
Moreover, since $R$ does not depend on the variance, $R(Z, \bx) = R(Z^*, \bx)$. 
Note that this example is highly simplified as an illustration, but the intuition that it provides holds more generally. 


\section{Technical details}\label{app:tech_details}

Recall that the Fisher information matrix $I(\genmodel) \in \mathbb{R}^{r \times r}$ is a positive semi-definite matrix, where the $(i,j)$-th entry is given by: 
\begin{align*}
	[I(\genmodel)]_{ij} = \mathbb{E}_{{\sf \mathbf{z}} \sim p_\bz (\cdot ; \genmodel)} \left[ \frac{\partial}{\partial \theta_i} \log p_\bz ({\sf \bz} ; \genmodel)  \frac{\partial}{\partial \theta_j} \log p_\bz ({\sf \bz} ; \genmodel)  \right]
\end{align*}

Recall that $Z$ is generated by drawing $m$ samples from $p_\bz (\cdot \hspace{1pt} ; \genmodel)$, where $\genmodel \in \parFam$. 
Recall further that $\cL: \cZ  \rightarrow \parFam$ denotes an estimator. 
An estimator $\cL$ is asymptotically normal and efficient if:
\begin{align}
	\sqrt{m} \left( \cL (Z) - \genmodel \right) \stackrel{d}{\rightarrow} \mathcal{N}( \mathbf{0}_r , I^{-1}(\genmodel) ) ,  \label{eq:MLE_asymp_norm}
\end{align}
as $m \rightarrow \infty$ for all $\genmodel \in \parFam$ where $I^{-1}(\genmodel)$ denotes the inverse of the Fisher information matrix at $\genmodel$.

Lastly,  let $\cP = \{ p_{\bz} ( \cdot \hspace{1pt} ; \genmodel ) : \genmodel \in \parFam \}$. 
The regularity conditions on $\cP$ that are discussed in Theorem \ref{thm:reg_hyp_test} are stated as follows. 
\begin{enumerate}
	\item $\parFam$ is a compact and open set of $\mathbb{R}^r$.\label{item:reg_cond_1} 
	\item Identifiability: $\bz \iid p_{\sf \bz}(\cdot; \genmodel)$ for $\genmodel \in \parFam$ and $\genmodel_1 \neq \genmodel_2$ implies $p_{\bz}(\cdot ; \genmodel_1)$ and $p_{\bz}(\cdot ; \genmodel_2)$ are distinct.
	\item 
	Common support: The support of $p_\bz(\cdot ; \genmodel)$ is independent of $\genmodel \in \parFam$. 
	\item Differentiability: All the second-order partial deriviates of $\log p_\bz(\bz ; \genmodel)$ with respect to $\genmodel$ exist and are continuous in $\genmodel$. \label{item:reg_cond_4}
	\item For any $\genmodel_0 \in \parFam$, there exists a neighborhood of $\genmodel_0$ and a function $\Pi(\bz)$, where $\mathbb{E}_{{\sf \bz} \sim p_\bz(\cdot ; \genmodel_0)} [ \Pi(\bz) ] < \infty$ and 
	\begin{align*}
		\qquad \qquad \left| \frac{\partial^2}{\partial \theta_i \partial \theta_j} \log p_\bz({\sf \bz} ; \genmodel ) \right| \leq \Pi(\bz) ,
	\end{align*}
	for all $\bz \in \mathcal{Z}$, all $\genmodel$ in the neighborhood of $\genmodel_0$, and $i,j \in [r]$.
	\item If $\genmodel^*$ is the data generating parameter: \label{item:reg_cond_6}
	\begin{enumerate}
		\item $\frac{\partial}{\partial \theta_i} \log p_\bz(\bz ; \genmodel^*)$ is square integrable for all $i \in [r]$. 
		\item $\mathbb{E}_{{\sf \bz} \sim p_\bz(\cdot ; \genmodel^*)} \left [ \frac{\partial}{\partial \theta_i} \log p_\bz({\sf \bz} ; \genmodel^*) \right] = 0$
		\item The Fisher information at $\genmodel^*$ satisfies:
		\begin{align*}
			\qquad \qquad [I(\genmodel^*)]_{ij}  
			&= \mathbb{E}_{{\sf \bz} \sim p_\bz(\cdot ; \genmodel^*)} \left [ \frac{\partial}{\partial \theta_i} \log p_\bz({\sf \bz} ; \genmodel^*)  \frac{\partial}{\partial \theta_j} \log p_\bz({\sf \bz} ; \genmodel^*) \right] 
			\\
			&= -\mathbb{E}_{{\sf \bz} \sim p_\bz(\cdot ; \genmodel^*)} \left [ \frac{\partial^2 }{\partial \theta_i \theta_j} \log p_\bz({\sf \bz} ; \genmodel^*) \right] 
		\end{align*}
		\item  Invertibility: Fisher information $I(\genmodel^*)$ at $\genmodel^*$  is positive-definite and invertible.
	\end{enumerate}
	\item Either all distributions in $\cP$ are lattice distributions on the same lattice or 
	each $p_\bz(\cdot ; \genmodel) \in \cP$ has a component such that, for a constant $k$ that is independent of $\genmodel$, 
	the $k$-fold convolution has a bounded density with respect to the Lebesgue measure. 	\label{item:reg_cond_7}
	\item For all $\genmodel \in \parFam$, there exists an unbiased estimator $\cL$ such that $\bbE [ | \cL (Z)  | ] < \infty$. 
	\label{item:reg_cond_8}
	\item $\parFam$ is a convex set. 
\label{item:reg_cond_9}
\end{enumerate}
There are variations on these regularity conditions, and we refer the reader to other works for further details \cite{bahadur1964fisher,lehmann2006theory,ly2017tutorial}. 
The compactness requirement in Condition \ref{item:reg_cond_1} and the continuity requirement in Condition \ref{item:reg_cond_4} ensure the existence of the MLE.
The remaining statements in Conditions \ref{item:reg_cond_1} through \ref{item:reg_cond_6} ensure the asymptotic normality of the MLE. 
Conditions \ref{item:reg_cond_7}-\ref{item:reg_cond_8} ensure the asymptotic normality of the MVUE (cf. \cite{portnoy1977asymptotic} for details).
Condition \ref{item:reg_cond_9} ensures the existence of $I((\genmodel + \genmodel')/2)$ for $\genmodel, \genmodel'  \in \parFam$, and this condition can be relaxed by providing a slightly different statement of Theorem \ref{thm:reg_hyp_test} (e.g., letting $\genmodel^*$ be $\genmodel(\bx)$ or $\genmodel(\bx')$).
\section{Proofs}\label{app:proofs}

\subsection{Theorem \ref{thm:reg_hyp_test}}

\begin{theorem} 
	Consider \eqref{eq:H_test}.
	Let $\genmodel^* = (\genmodel(\fvec)  + \genmodel (\fvec')) / 2$. 
	Suppose that $\bz_i$ and $\bz'_i$ are drawn i.i.d. from $p(\cdot \hspace{1pt} ; \genmodel (\fvec) )$ and $p(\cdot \hspace{1pt} ; \genmodel (\fvec') )$, respectively, for all $i \in [m]$
	and $\cP = \{ p_{\bz} ( \cdot \hspace{1pt} ; \genmodel ) : \genmodel \in \parFam \}$ is a regular exponential family that meets the regularity conditions stated in Appendix \ref{app:tech_details}.
	If $\hat{H}$ is defined as:
	\begin{align*}
		\hat{H} = H_1 
		\iff
		(\cL^+(Z)  - \cL^+(Z'))^\top I ( \genmodel^* ) (\cL^+(Z)  - \cL^+(Z'))  \geq \frac{2}{m} \chi_r^2 ( 1 - \epsilon ) , 
	\end{align*}
	then $\bbP( \hat{H} = H_{1} | H = H_{0} ) \leq \epsilon$ as $m \rightarrow \infty$. 
	If $r = 1$, then $\hat{H}$ is the UMPU test as $m \rightarrow \infty$, i.e., $\lim_{m \rightarrow \infty} \bbP(\hat{H} = H_0 | H = H_1 ) = \alpha^*(\epsilon)$. 
\end{theorem}

\begin{proof}

	The regularity conditions required for Theorem \ref{thm:reg_hyp_test} are stated in Appendix \ref{app:tech_details}. 
	The definition of asymptotic normality and efficiency is also given in Appendix \ref{app:tech_details}.

	Under the regularity conditions, 
	we have three results. 
	First, 
	under Conditions  \ref{item:reg_cond_1} and \ref{item:reg_cond_4}, 
	the MLE exists and, from Conditions \ref{item:reg_cond_1}-\ref{item:reg_cond_6}, 
	it is asymptotically normal and efficient \cite{bahadur1964fisher,lehmann2006theory,ly2017tutorial}.
	Second,
	under Conditions \ref{item:reg_cond_1}-\ref{item:reg_cond_8}, 
	the MVUE exists and is also asymptotically normal and efficient \cite{portnoy1977asymptotic}.
	Third, Condition \ref{item:reg_cond_9} ensures the existence of $I((\genmodel + \genmodel')/2)$ for $\genmodel, \genmodel'  \in \parFam$, and this condition can be relaxed by providing a slightly different statement of Theorem \ref{thm:reg_hyp_test} (e.g., letting $\genmodel^*$ be $\genmodel(\bx)$ or $\genmodel(\bx')$).
	
	By the second result,
	\begin{align*}
		\sqrt{m}( \cL^+ (Z)  - \genmodel(\fvec) ) \stackrel{d}{\rightarrow} \mathcal{N}( \mathbf{0}_r , I^{-1} (\genmodel (\fvec) )  )
	\end{align*}
	 as $m \rightarrow \infty$, where $\bz_i \iid p(\cdot ; \genmodel)$ and $Z = (\bz_1, \hdots, \bz_m)$.
	Therefore, as $m \rightarrow \infty$, 
	\begin{align}
		\sqrt{m}( \cL^+ (Z)  - \genmodel(\fvec) - \cL^+ (Z') + \genmodel( \fvec' ) ) \stackrel{d}{\rightarrow} \mathcal{N}( \mathbf{0}_r , I^{-1} (\genmodel ( \fvec ) )  + I^{-1} (\genmodel ( \fvec' ) ) ) \label{eq:asym_norm_diff}
	\end{align}

	Recall the hypothesis test \eqref{eq:H_test} from Section \ref{sec:formal_goal_H_test}.
	When $H = H_0$, 
	$\genmodel(\fvec) = \genmodel(\fvec') = \genmodel^*$. 
	Therefore, by \eqref{eq:asym_norm_diff},
	\begin{align}
		\sqrt{m}( \cL^+ (Z)  - \cL^+ (Z') ) \stackrel{d}{\rightarrow} \mathcal{N}( \mathbf{0}_r , 2 I^{-1} (\genmodel^*  )  ) \label{eq:asym_norm_diff_F0}
	\end{align}
	as $m \rightarrow \infty$ when $H = H_0$, 
	which implies that, as $m \rightarrow \infty$, 
	the  two-sample, two-sided hypothesis test in \eqref{eq:H_test}
	becomes a two-sample, one-sided test of on the mean of a multivariate Gaussian random variable. 
	Under \eqref{eq:asym_norm_diff_F0}, 
	\begin{align*}
		(\cL^+(Z)  - \cL^+(Z'))^\top I ( \genmodel^* ) (\cL^+(Z)  - \cL^+(Z'))  \sim \frac{2}{m} \chi_r^2
	\end{align*}
	Therefore, if $\hat{H}$ satisfies:
	\begin{align}
		\hat{H} = H_1 
		\iff
		(\cL^+(Z)  - \cL^+(Z'))^\top I ( \genmodel^* ) (\cL^+(Z)  - \cL^+(Z'))  \geq \frac{2}{m} \chi_r^2 ( 1 - \epsilon ) \label{eq:F_rule}
	\end{align}
	then $\hat{H}$ has a FPR $\leq \epsilon$, as desired.  
	
	Although $\hat{H}$ is not necessarily the UMPU for $r > 1$, 
	it is well known that it is the UMPU test of size $\epsilon$ for the univariate Gaussian case, i.e., when $r = 1$ (cf. Section 8.3 of \cite{casella2021statistical}).

	One may have noticed that the hypothesis test in \eqref{eq:F_rule}  (and \eqref{eq:thm_H1_def}) uses $Z$ and $Z'$ to choose between $H_0$ and $H_1$, 
	whereas in the original problem statement in Section \ref{sec:problem_statement}, 
	the hypothesis test uses $\cD$ and $\cD'$. 
	In other words, decision robustness requires that one cannot determine whether $\bx \neq \bx'$---or, equivalently, $\genmodel(\bx) \neq \genmodel(\bx')$---from $\cD$ and $\cD'$ for any $\cQ$. 
	Although $\cD$, $\cD'$, and $\cQ$ do not appear in the analysis above, 
	the test in \eqref{eq:F_rule} ensures (approximate asymptotic) decision robustness by designing a test that works for all $\cQ$ and consequent decisions $\cD$ and $\cD'$. 
	
	To see this,  we first note that, if expressed as a Markov chain, the random variables of interest would be written as $\bx \rightarrow \genmodel \rightarrow Z \rightarrow \cD$. 
	By the data processing inequality, 
	any test that uses $Z$ is stronger (i.e., has a higher TPR and lower FPR) than the corresponding test using $\cD$.
	Intuitively, since $\cD$ is determine by $Z$ and $\cD'$ from $Z'$, 
	if one cannot determine whether $\bx \neq \bx$ (or $\genmodel(\bx) \neq \genmodel(\bx')$) from $Z$ and $Z'$, 
	then one cannot do any better given $\cD$ and $\cD'$ for any $\cQ$. 
	We can therefore conclude that the guarantees of \eqref{eq:F_rule} hold for all $\cQ$ and the original hypothesis test in \eqref{eq:H_test}. 
\end{proof}

Note that the regularity conditions required in Theorem \ref{thm:reg_hyp_test} are fairly mild. 
Recall that exponential families capture a broad class of distributions of interest. 
In particular, they are the only families of distributions that have finite-dimensional sufficient statistics, and a distribution almost always belongs to an exponential family if it has a conjugate prior. 
Regular exponential families are canonical exponential families if the natural parameter space is an open set in $\parFam$. 
The remaining regularity conditions are common and often implicitly assumed in discussions of the MVUE or MLE.

\textbf{Alternate result}.
We would like to remark that the audit could be modified to use the MLE instead of the MVUE. 
Using the MLE would provide the same guarantee as Theorem \ref{thm:reg_hyp_test}. 
In fact, it would require fewer conditions, as follows. 
\begin{theoremapp} 
	Consider \eqref{eq:H_test}.
	Let $\genmodel^* = (\genmodel(\fvec)  + \genmodel (\fvec')) / 2$. 
	Suppose that $\bz_i$ and $\bz'_i$ are drawn i.i.d. from $p(\cdot \hspace{1pt} ; \genmodel (\fvec) )$ and $p(\cdot \hspace{1pt} ; \genmodel (\fvec') )$, respectively, for all $i \in [m]$
	and $\cP = \{ p_{\bz} ( \cdot \hspace{1pt} ; \genmodel ) : \genmodel \in \parFam \}$ meets Conditions\ref{item:reg_cond_1} through \ref{item:reg_cond_6} stated in Appendix \ref{app:tech_details}.
	If $\hat{H}$ is defined as:
	\begin{align}
		\hat{H} = H_1 
		\iff
		(\cL^+(Z)  - \cL^+(Z'))^\top I ( \genmodel^* ) (\cL^+(Z)  - \cL^+(Z'))  \geq \frac{2}{m} \chi_r^2 ( 1 - \epsilon ) , 
	\end{align}
	then $\bbP( \hat{H} = H_{1} | H = H_{0} ) \leq \epsilon$ as $m \rightarrow \infty$. 
	If $r = 1$, then $\hat{H}$ is the UMPU test as $m \rightarrow \infty$. 
\end{theoremapp}
Therefore, one may wish to use the MLE instead of the MVUE in Algorithm \ref{alg:test}. 
The reason that one may wish to use the MVUE is because it provides another guarantee (cf. Proposition  \ref{prop:mvue}).

\subsection{Proposition \ref{prop:mvue}}

\begin{proposition}
	Consider \eqref{eq:user_dm}. 
	Let $G\in \{ G_0 , G_1 \}$ denote the true (unknown) hypothesis.
	Suppose that $v_0 , v_1 : \parFam \rightarrow \bbR$ are affine mappings
	and there exists $\bu : \bbR^n \rightarrow \mathcal{U}$ such that, for $\cF(\bx) = \{ \bz_1, \hdots, \bz_m \}$, one can write $\bu(\bz_i) \sim p_{\bu}( \cdot ; \hspace{1pt} v_1(\genmodel(\bx)) - v_0( \genmodel(\bx) ))$ for all $i \in [m]$.
	Then, if the UMP test with a maximum FPR of $\rho$ exists, 
	it is given by the following decision rule: reject $G_0$ 
	(choose $A_1$) 
	when the minimum-variance unbiased estimate $\usermodel = \cL^+ (Z)$ satisfies $v_1( \usermodel )  - v_0 (\usermodel ) > \eta_\rho$ where $\bbP( v_1 ( \usermodel ) - v_0 ( \usermodel  ) > \eta_\rho | G= G_0 ) = \rho$;
	otherwise, accept $G_0$ (choose $A_0$). 
\end{proposition}
\begin{proof}
	We begin with a result from Ghobadzadeh et al. \cite{ghobadzadeh2009role}. 
	
	\begin{lemmaApp}[\citep{ghobadzadeh2009role}, Theorem 1]\label{lem:UMP}
		Consider a one-sided binary composite hypothesis test of $\bar{G}_0: \mathbf{w} \sim p_{\sf \mathbf{w} }( \cdot ; \gamma), \gamma \leq \gamma_b$ against $\bar{G}_1: \mathbf{w} \sim p_{\sf \mathbf{w} }(  \cdot ; \gamma ),  \gamma > \gamma_b$, where $\gamma_b$ is known. Let $\Gamma_0 = \{ \gamma : \gamma \leq \gamma_b\}$, $\Gamma_1 = \{ \gamma : \gamma > \gamma_b\}$, and $\Gamma = \Gamma_0 \cup \Gamma_1$. Let $\rho$ be the maximum allowable false positive rate. If the uniformly most powerful (UMP) test  exists, then it is defined by the following decision rule: reject $\bar{G}_0$ when the minimum variance unbiased estimator (MVUE) of $\gamma \in \Gamma$, denoted by $\tilde{\gamma}^+$, satisfies $\tilde{\gamma}^+ > \gamma_\rho$, where $P( \tilde{\gamma}^+ > \gamma_\rho | H = \bar{G}_0) = \rho$. 
	\end{lemmaApp}
	
	Our result follows directly from two observations. 
	First, since $v_0, v_1$ are affine, if $\usermodel$ is the MVUE of $\genmodel$, 
	then $v_1(\usermodel) - v_0(\usermodel)$ is also the MVUE of $v_1 (\genmodel) - v_0(\genmodel)$.
	Second, our setting is equivalent to that in Lemma \ref{lem:UMP} with the substitutions 
	$\gamma = v_1 (\genmodel) - v_0(\genmodel)$, $\gamma_b = 0$, and
	$\bw = \bu(\bz)$. 
\end{proof}

\subsection{Theorem \ref{thm:cost_of_reg}}

Recall from footnote \ref{ft:static_R} that $R$ can be time-varying. 
For example,, the platform's revenue sources may change with time.
As long as the time-varying objective function $R^\tindex$ satisfies the conditions in Theorem \ref{thm:cost_of_reg} at every time step, then the result holds unchanged at every time step.

\begin{theorem}
	Suppose there exists $\Omega \subset [r]$ where $1 < |\Omega| < r$
	such that
	$R(\genmodel_1, \bx) = R(\genmodel_2, \bx)$ if $\theta_{1,i} = \theta_{2,i}$ for all $i \notin \Omega$.
	Suppose that, for any $\genmodel \in \parFam$, $\beta > 0$ and $\bv \in \bbR^r$, 
	there exist a vector $\bar{\genmodel}_{\Omega}$ where $\bar{\theta}_{\Omega,i} = 0$ for all $i \notin \Omega$
	and a constant $\kappa > 0$ such that $\bv^\top I ( \genmodel + \kappa \bar{\genmodel}_{\Omega} ) \bv < \beta$ and $\genmodel + \kappa \bar{\genmodel}_{\Omega} \in \parFam$. 
	Then, if $m < \infty$, 
	there exists a $\cZ$ such that the cost of regulation for $\bx$ under Algorithm \ref{alg:test} is $0$. 
\end{theorem}
\begin{proof}
	Let $Z_* \in  \arg \max_{W \in \cZ} R(W , \bx)$ be a reward-maximizing feed and $\usermodel_{*} = \cL^+(Z_* )$. 
	Similarly, let $Z_*' \in  \arg \max_{W \in \cZ} R(W , \bx')$ and $\usermodel{}'_{*} = \cL^+( Z'_{*} )$. 
	Recall that, under the statement conditions, there exists a vector $\bar{\genmodel}_{\Omega}$ where $\bar{\theta}_{\Omega,i} = 0$ for all $i \notin \Omega$
	and a constant $\kappa > 0$ such that $\bv^\top I ( \genmodel + \kappa \bar{\genmodel}_{\Omega} ) \bv < \beta$ and $\genmodel + \kappa \bar{\genmodel}_{\Omega} \in \parFam$ for any $\genmodel \in \parFam$, $\beta > 0$, and $\bv \in \bbR^r$. 
	Let us take
	$\genmodel =  \usermodel_{*} = \cL^+( Z_* )$,
	$\beta =  \frac{2}{m} \chi^2_r ( 1 - \epsilon )$,
	and $\bv = \usermodel{}'_{*} - \usermodel_{*}$.
	Finally, let $\usermodel = \usermodel_{*} + \bar{\genmodel}_{\Omega}$,
	where $\bar{\genmodel}_{\Omega}$ is defined as given in the statement.
	Then, 
	\begin{align*}
		(\usermodel{}'_{*} - \usermodel_{*} )^\top I (  \usermodel_{*} + \bar{\genmodel}_{\Omega} ) (\usermodel{}'_{*} -  \usermodel_{*} )
		< \frac{2}{m} \chi^2_r ( 1 - \epsilon ) .
	\end{align*}
	Letting $\usermodel{}'  = \usermodel{}'_{*}  + \bar{\genmodel}_{\Omega}$ and recalling $\usermodel = \usermodel_{*} + \bar{\genmodel}_{\Omega}$ gives
	\begin{align*}
		(\usermodel' - \usermodel )^\top I (  \usermodel ) (\usermodel{}' -  \usermodel )
		< \frac{2}{m} \chi^2_r ( 1 - \epsilon ) ,
	\end{align*}
	which implies that, 
	as long as $\cZ$ is large enough such that there exist $Z$ and $Z'$ such that $\cL^+(Z)  = \usermodel$  and  $\cL^+(Z')  = \usermodel{}'$,
	then both $Z$ and $Z'$ comply with the regulation.
	In other words, as long as $\cZ$ contains content that is expressive enough, we know that $Z , Z' \in \cZ(\cS)$.
	
	It remains to show that the  cost of regulation is $0$. 
	To do so, we show that $Z$, which is in the feasible set, 
	achieves the maximum reward $\max_{W \in \cZ} R(W , \bx) = R(\usermodel_{*} , \bx)$. 
	That the cost of regulation is $0$ follows from the fact that $R (\usermodel_{*}  , \bx) = R(\usermodel_{*}  + \bar{\genmodel}_{\Omega}  , \bx )$ because $\bar{\theta}_{\Omega,i} = 0$ for $i \notin \Omega$.
\end{proof}

\section{Additional discussion} \label{app:discussion}

\subsection{Two remarks}\label{app:relaxing_assumptions}

Recall that we made two simplifications in the main text for readability. 
As noted in the main text, these simplifications do not change our main findings, 
which implies that our results hold under conditions more general than those given in the main text. 

First, recall from footnote \ref{ft:static_R} that $R$ can be time-varying, i.e., let the objective function that the platform wishes to maximize at time step be given by $R^\tindex$. 
This allows our analysis to accommodate settings in which the platform's objectives (e.g., revenue sources) change with time. 
Allowing the objective function to vary in time does not change our conclusions. 
Our findings with respect to $R$ appear in Theorem \ref{thm:cost_of_reg} and the discussion that follows. 
As apparent in the proof of Theorem \ref{thm:cost_of_reg}, 
adopting a time-varying $R^\tindex$ leads to the same result as long as $R^\tindex$ satisfies the conditions in the theorem statement at the time step $t$ of interest. 

Second, recall from footnote \ref{ft:more_info_estimator} that a user's beliefs are often influenced by information other than their feed $Z$. 
For instance, a user's beliefs may depend on conversations that they have offline or on their previous beliefs.
One could incorporate this into our setup by letting $J \in \cJ$ denote any information other than $Z$ that the user uses to form or update their beliefs
and the estimator $\cL : \cZ \times \cJ \rightarrow \parFam$ denote the user's learning behavior (that incorporates information both on and off the platform) such that the user's belief after viewing $Z$ and observing information $J$ is given by $\cL ( Z , J )$. 

The outside information $J$ does not affect our results because it can be absorbed into $\cL$. 
That is, because we are only interested in how the user's beliefs are affected when the user is shown $Z'$ instead of $Z$, and vice versa, $J$ can effectively be ignored. 
Another way to see this is by recalling the definition of decision robustness. 
Decision robustness supposes that there are two identical (hypothetical) users, one of whom is shown $Z$ and the other $Z'$. 
For the purpose of comparing the outcomes under $Z$ and $Z'$, 
$J$ could be treated as part of the original identical users.

\subsection{Impact and consequences}

Our hope is that this work can contribute to the ongoing conversations about social media and its regulation. 
In light of the difficulties in designing and enforcing a regulation, we focus on the latter half of the process
by proposing an auditing procedure.
Auditing social media remains a challenging topic because changes to the ecosystem can have far-reaching consequences. 
As such, we sought to consider the various stakeholders in the system. 

In particular, we studied our framework from the perspectives of the auditor, platform, and user. 
The proposed test focuses on a given pair of inputs. 
We made this choice intentionally to prevent issues that often arise when a regulatory test focuses only on average behavior, which can sometimes result in good outcomes for most individuals but unsatisfactory outcomes for a small subset of the population (i.e., a minority group). 
We considered the perspective of the auditor by acknowledging the difficulties in designing regulations that are enduring and adaptable while also being precise and implementable. 
To this end, our main contribution is a test that translates counterfactual regulations into a principled regulatory procedure. 
We also consider the platform's perspective by studying how the audit affects the platform's ability to maximize some objective function (e.g., revenue, user engagement, a combination of these factors, and more). 
This discussion returns to the user's perspective by examining how the audit changes the feed that the platform is incentivized to show users with particular attention to the content diversity of users' feeds. 

To the best of our abilities, we attempt to acknowledge and address the impact of our work by considering various perspectives of our proposal, explicitly mentioning what problems are within the scope of this work and pointing to appropriate references. However, there may be angles that we have missed. There is also the potential to misuse the proposed framework. For instance, if a platform decided to adopt our procedure as a self-regulatory measure, the outcome would depend on how seriously the platform engages in conversations on designing the counterfactual feeds. 
Another potential misuse would be adversarially designing the features that represent content such that the regulation is ineffective. However, a good-faith effort to choose and test these features appropriately should resolve this issue. 
One might also be concerned with user privacy. 
In response, we provide several comments. 
First, the proposed test does not require user-specific information. 
Second, the inputs $(\bx, \bx')$ need not represent real users, 
and we would in fact recommend that they correspond to hypothetical users. 
If the audit uses hypothetical users, then the main way that user information is revealed to the auditor is via the content that appears in the feeds that the auditor uses to audit because much of the content on social media is generated by users themselves. 
Although this issue seems unavoidable, 
one encouraging feature of the audit is that the auditor only requires access to the feature vectors (or embedding) of each piece of content.
Therefore, as long as the auditor has no intention of unmasking the identity of users, 
the test could be run over these features,
and this data could be immediately discarded afterwards. 
The only output that would be preserved is the outcome of the audit. 
Lastly, one implicit source of bias could be in the selection of the model family $\Theta$, which is a decision made by the auditor. 
We choose to leave $\Theta$ unspecified because doing so means that our analysis can be generalized to any $\parFam$ of interest. However, the auditor should test different choices of $\parFam$ and observe the outcomes. 
Recall that $\parFam$ captures the set of possible generative models (or, in the context of Section \ref{sec:insight}, possible cognitive models).
In choosing the model family $\parFam$, a simple $\parFam$ is more tractable and interpretable while a complex $\parFam$ is more general.

\end{document}